\documentclass[aip,jcp,preprint]{revtex4-1}

\usepackage{amsmath,amssymb,amsthm,amscd}
\usepackage[pdftex]{graphicx,color}
\usepackage[mathscr]{eucal}
\usepackage{enumitem}
\usepackage{verbatim}
\usepackage[roman]{sublabel}
\usepackage{ulem}

\usepackage{afterpage}
\usepackage{comment}

\definecolor{link_red}{rgb}{0.7,0,0}
\definecolor{cite_blue}{rgb}{0,0,0.97}
\definecolor{cadmium}{RGB}{255,97,3}

\newcommand{\leqs}{\leqslant}
\newcommand{\geqs}{\geqslant}

\newcommand{\defas}{\mathrel{\raise.095ex\hbox{$:$}\mkern-4.2mu=}}
\newcommand{\defasr}{\mathrel{=\mkern-4.2mu\raise.095ex\hbox{$:$}}}

\newcommand{\vep}{\varepsilon}

\newcommand{\mbi}[1]{\boldsymbol{\mathit{#1}}}

\newcommand{\bsym}[1]{\boldsymbol{#1}}

\newcommand{\vabs}[1]{\left\lvert #1 \right\rvert}

\newcommand{\vnorm}[1]{\left\lVert #1 \right\rVert}

\newcommand{\set}[1]{\{ #1 \}}
\newcommand{\vset}[1]{\left\{ #1 \right\}}

\newcommand{\mbb}[1]{\mathbb{#1}}
\newcommand{\mrm}[1]{\mathrm{#1}}
\newcommand{\msf}[1]{\mathsf{#1}}
\newcommand{\mscr}[1]{\mathscr{#1}}
\newcommand{\mcal}[1]{\mathcal{#1}}

\newcommand{\pref}[1]{(\ref{#1})}

\DeclareMathOperator{\supp}{supp}

\makeatletter
\newcommand*{\da@rightarrow}{\mathchar"0\hexnumber@\symAMSa 4B }
\newcommand*{\da@leftarrow}{\mathchar"0\hexnumber@\symAMSa 4C }
\newcommand*{\xdashrightarrow}[2][]{%
  \mathrel{%
    \mathpalette{\da@xarrow{#1}{#2}{}\da@rightarrow{\,}{}}{}%
  }%
}
\newcommand{\xdashleftarrow}[2][]{%
  \mathrel{%
    \mathpalette{\da@xarrow{#1}{#2}\da@leftarrow{}{}{\,}}{}%
  }%
}
\newcommand*{\da@xarrow}[7]{%
  \sbox0{$\ifx#7\scriptstyle\scriptscriptstyle\else\scriptstyle\fi#5#1#6\m@th$}%
  \sbox2{$\ifx#7\scriptstyle\scriptscriptstyle\else\scriptstyle\fi#5#2#6\m@th$}%
  \sbox4{$#7\dabar@\m@th$}%
  \dimen@=\wd0 %
  \ifdim\wd2 >\dimen@
    \dimen@=\wd2 %
  \fi
  \count@=2 %
  \def\da@bars{\dabar@\dabar@}%
  \@whiledim\count@\wd4<\dimen@\do{%
    \advance\count@\@ne
    \expandafter\def\expandafter\da@bars\expandafter{%
      \da@bars
      \dabar@ 
    }%
  }%
  \mathrel{#3}%
  \mathrel{%
    \mathop{\da@bars}\limits
    \ifx\\#1\\%
    \else
      _{\copy0}%
    \fi
    \ifx\\#2\\%
    \else
      ^{\copy2}%
    \fi
  }%
  \mathrel{#4}%
}
\makeatother

\theoremstyle{plain}
\newtheorem{theorem}{Theorem}[section]
\newtheorem{corollary}[theorem]{Corollary}
\newtheorem{lemma}[theorem]{Lemma}
\newtheorem{proposition}[theorem]{Proposition}

\newtheorem{m_theorem}{Theorem}

\theoremstyle{definition}
\newtheorem{definition}[theorem]{Definition}

\DeclareMathOperator{\Poisson}{Poisn}

\begin{document}

\title{Modeling delay in genetic networks: From delay birth-death processes to delay stochastic differential equations}

\author{Chinmaya Gupta} \affiliation{Department of Mathematics, University of Houston, Houston, TX}
\author{Jos\'e Manuel L\'opez} \affiliation{Department of Mathematics, University of Houston, Houston, TX}
\author{Robert Azencott} \affiliation{Department of Mathematics, University of Houston, Houston, TX}
\author{Matthew R.\ Bennett} \affiliation{Department of Biochemistry \& Cell Biology, Rice University, Houston, TX}\affiliation{Institute of Biosciences \& Bioengineering, Rice University, Houston, TX}
\author{Kre\v{s}imir Josi\'{c}} \affiliation{Department of Mathematics, University of Houston, Houston, TX}\affiliation{Department of Biology \& Biochemistry, University of Houston, Houston, TX}
\author{William Ott} \affiliation{Department of Mathematics, University of Houston, Houston, TX}

\begin{abstract}
Delay is an important and ubiquitous aspect of many biochemical processes.  For example, delay plays a central role in the dynamics of genetic regulatory networks as it stems from the sequential assembly of first mRNA and then protein.  Genetic regulatory networks are therefore frequently modeled as stochastic birth-death processes with delay.  Here we examine the relationship between delay birth-death processes and their appropriate approximating delay chemical Langevin equations. We prove that the distance between these two descriptions, as measured by expectations of functionals of the processes, converges to zero with increasing system size. Further, we prove that the delay birth-death process converges to the thermodynamic limit as system size tends to infinity.  Our results hold for both fixed delay and distributed delay.  Simulations demonstrate that the delay chemical Langevin approximation is accurate even at moderate system sizes.  It captures dynamical features such as the spatial and temporal distributions of transition pathways in metastable systems, oscillatory behavior in negative feedback circuits, and cross-correlations between nodes in a network. Overall, these results provide a foundation for using delay stochastic differential equations to approximate the dynamics of birth-death processes with delay.
\end{abstract}

\maketitle

\section{Introduction}
\label{s:intro}

\noindent
Gene regulatory networks play a central role in cellular function by translating genotype into phenotype. By dynamically controlling gene expression, gene regulatory networks provide cells with a mechanism for responding to environmental challenges. Therefore, creating accurate mathematical models of gene regulation is a central goal of mathematical biology. 

Delay in protein production can significantly affect the dynamics of gene regulatory networks. For example, delay can induce oscillations in systems with negative feedback~\cite{Amir_Meshner_Beatus_Stavans_2010, Bratsun_Volfson_Tsimring_Hasty_2005, Goodwin_1965, Lewis_2003, Mather_Bennett_Hasty_Tsimring_2009, Monk_2003, Smolen_Baxter_Byrne_1999}, and has been implicated in the production of robust, tunable oscillations in synthetic gene circuits containing linked positive and negative feedback~\cite{Stricker_Cookson_Bennett_Mather_Tsimring_Hasty_2008, Tigges_Marquez-Lago_Stelling_Fussenegger_2009}.  Indeed, delayed negative feedback is thought to govern the dynamics of circadian oscillators~\cite{Smolen_Baxter_Byrne_2002, Sriram_Gopinathan_2004}, a hypothesis experimentally verified in mammalian cells~\cite{Ukai-Tadenuma_etal_2011}.

In genetic regulatory networks, noise and delay interact in subtle and complex ways.  Delay can affect the stochastic properties of gene expression and hence the phenotype of the cell~\cite{Bratsun_Volfson_Tsimring_Hasty_2005, Gronlund_Lotstedt_Elf_2010, Gronlund_Lotstedt_Elf_2011, Maithreye_Sarkar_Parnaik_Sinha_2008, Scott_2009}.  It is well known that noise can induce switching in bistable genetic circuits~\cite{XiaoWang_PNAS_2013, hong:2012, he:2011, Ozbudak2004, Kepler_Elston_2001, aurell:2002, warren:2005, balaban:2004, gardner:2000, nevozhay:2012}; the infusion of delay dramatically enhances the stability of such circuits~\cite{Gupta_Lopez_Ott_Josic_Bennett_2013} and can induce an analog of stochastic resonance~\cite{Fischer_Imkeller_2005, Fischer_Imkeller_2006}.  Variability in the delay time (distributed delay) can accelerate signaling in transcriptional signaling cascades~\cite{Josic_Lopez_Ott_Shiau_Bennett_2011}.

Given the importance of delay in gene regulatory networks, it is necessary to develop methods to simulate and analyze such systems across spatial scales.  In the absence of delay, it is well known that chemical reaction networks are accurately modeled by ordinary differential equations (ODEs) in the thermodynamic limit, {\itshape i.e.} when molecule numbers are sufficiently large. When molecule numbers are small, however, stochastic effects can dominate.  In this case, the chemical master equation (CME) describes the evolution of the probability density function over all states of the system. Gillespie's stochastic simulation algorithm (SSA)~\cite{Gillespie_1977} samples trajectories from the probability distribution described by the CME.

While exact, the CME is difficult to analyze and the SSA can be computationally expensive.  To address these issues, a hierarchy of coarse-grained approximations of the SSA has been developed~\cite{Higham_2008_survey} (see Figure \ref{f:hierarchy}). Spatially discrete approximations, such as  $\tau $-leaping~\cite{Cao_Gillespie_Petzold_2006, Cao_Gillespie_Petzold_2007, Gillespie_2001, Xu_Cai_2008} and $K$-leaping~\cite{Cai_Xu_2007} trade exactness for efficiency.   At the next level are chemical Langevin equations (CLEs), which are stochastic differential equations of dimension equal to the number of species in the biochemical system.  CLEs offer two advantages.  First, unlike the SSA, the well-developed ideas from random dynamical systems and stochastic differential equations apply to CLEs.  Second, it is straightforward to simulate large systems using CLEs.  Finally, in the thermodynamic limit, one arrives at the end of the Markovian hierarchy: the reaction rate equation (RRE).

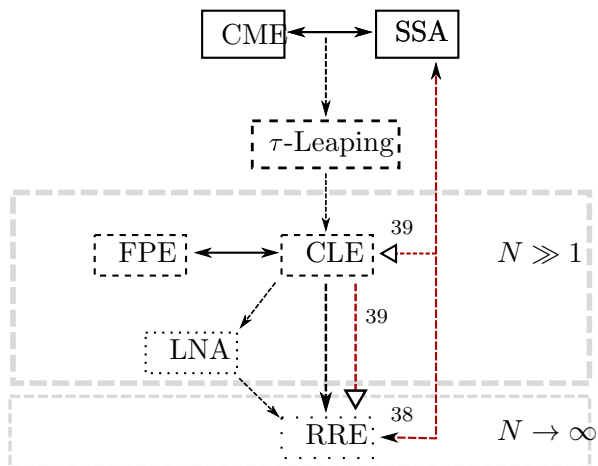
\begin{figure}
\begin{center}
\hspace*{-0.25cm}
\begingroup%
  \makeatletter%
  \providecommand\color[2][]{%
    \errmessage{(Inkscape) Color is used for the text in Inkscape, but the package 'color.sty' is not loaded}%
    \renewcommand\color[2][]{}%
  }%
  \providecommand\transparent[1]{%
    \errmessage{(Inkscape) Transparency is used (non-zero) for the text in Inkscape, but the package 'transparent.sty' is not loaded}%
    \renewcommand\transparent[1]{}%
  }%
  \providecommand\rotatebox[2]{#2}%
  \ifx\svgwidth\undefined%
    \setlength{\unitlength}{256.31220703bp}%
    \ifx\svgscale\undefined%
      \relax%
    \else%
      \setlength{\unitlength}{\unitlength * \real{\svgscale}}%
    \fi%
  \else%
    \setlength{\unitlength}{\svgwidth}%
  \fi%
  \global\let\svgwidth\undefined%
  \global\let\svgscale\undefined%
  \makeatother%
  \begin{picture}(1,0.75770602)%
    \put(0,0){\includegraphics[width=\unitlength]{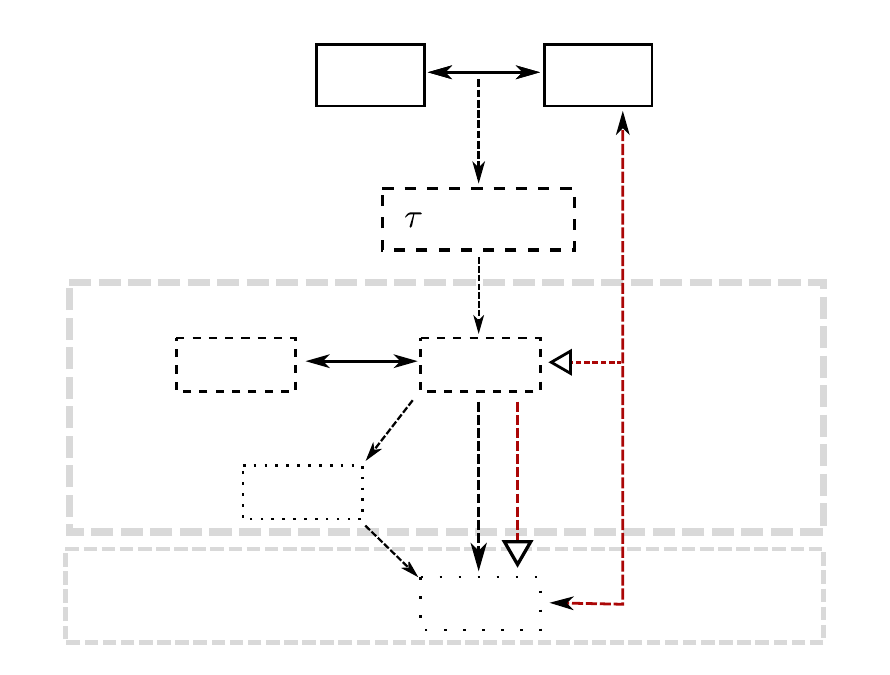}}%
    \put(0.63978037,0.66330034){\color[rgb]{0,0,0}\makebox(0,0)[lb]{\smash{SSA}}}%
    \put(0.63353796,0.08899866){\color[rgb]{0,0,0}\makebox(0,0)[lb]{\smash{\cite{Schlicht_Winkler_2008}}}}%
    \put(0.5960835,0.23257399){\color[rgb]{0,0,0}\makebox(0,0)[lb]{\smash{\cite{Brett_Galla_2013}}}}%
    \put(0.63353796,0.36366461){\color[rgb]{0,0,0}\makebox(0,0)[lb]{\smash{\cite{Brett_Galla_2013}}}}%
    \put(0.38384251,0.65966796){\color[rgb]{0,0,0}\makebox(0,0)[lb]{\smash{CME}}}%
    \put(0.47656487,0.50123945){\color[rgb]{0,0,0}\makebox(0,0)[lb]{\smash{-Leaping}}}%
    \put(0.50777693,0.33893741){\color[rgb]{0,0,0}\makebox(0,0)[lb]{\smash{CLE}}}%
    \put(0.2331119,0.33893741){\color[rgb]{0,0,0}\makebox(0,0)[lb]{\smash{FPE}}}%
    \put(0.30802054,0.19536251){\color[rgb]{0,0,0}\makebox(0,0)[lb]{\smash{LNA}}}%
    \put(0.50777688,0.07051479){\color[rgb]{0,0,0}\makebox(0,0)[lb]{\smash{RRE}}}%
    \put(0.63978037,0.66330034){\color[rgb]{0,0,0}\makebox(0,0)[lb]{\smash{SSA}}}%
    \put(0.78959824,0.33869497){\color[rgb]{0,0,0}\makebox(0,0)[lb]{\smash{$N \gg 1$}}}%
    \put(0.78959824,0.07651419){\color[rgb]{0,0,0}\makebox(0,0)[lb]{\smash{$N \to \infty$}}}%
  \end{picture}%
\endgroup%
\end{center}
\caption{Schematic of the modeling hierarchy for biochemical systems. The black arrows link the various components of the theory for Markov systems (no delay); red arrows link the corresponding delay components.  Numbers attached to arrows refer to papers that establish the corresponding links.  Empty arrowheads denote heuristic derivations; in this paper, we rigorously establish the dSSA to dCLE and dCLE to dRRE links.
}\label{f:hierarchy}
\end{figure}

\begin{figure}
\includegraphics[]{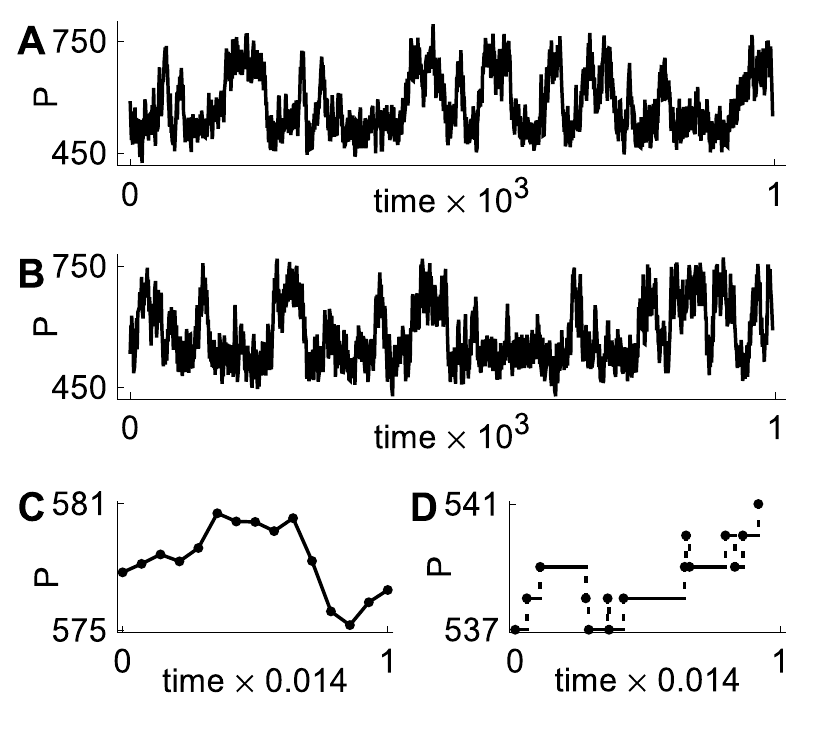}
\caption{(A) Typical trace obtained by using the dCLE.  (B) Corresponding trace obtained using  dSSA. By zooming into a small time segment, we see that the trace obtained from the dCLE (C) consists of equi-spaced time points (corresponding to an Euler discretization) and is continuous. The corresponding segment of the trace from the dSSA (D) consists of events that occur at random times and has jump discontinuities. The displayed traces were gathered after a long transient. The model simulated is the single-gene positive feedback model with $N = 500$; see Section \ref{sec:pf}.}\label{fig:intro}
\end{figure}

The Markovian hierarchy above (no delay) is well-understood~\cite{Gillespie_Hellander_Petzold_2013, Higham_2008_survey}, but a complete analogue of the Markovian theory does not yet exist for systems with delay.  The SSA has been generalized to a delay version - the dSSA - to allow for both fixed \cite{Barrio_Burrage_Leier_Tian_2006, Bratsun_Volfson_Tsimring_Hasty_2005}
and variable \cite{Josic_Lopez_Ott_Shiau_Bennett_2011, Schlicht_Winkler_2008} delay.  Some analogues of $\tau$-leaping exist for systems with delay; see {\itshape e.g.} $D$-leaping \cite{bayati2009d}.

Several methods have been used to formally derive a delay chemical Langevin equation (dCLE) from the delay chemical master equation (dCME); see Section~\ref{sec:discuss} for details.  Brett and Galla~\cite{Brett_Galla_2013} use the path integral formalism of Martin, Siggia, Rose, Janssen, and de Dominicis to derive a dCLE approximation without relying on a master equation.  The Brett and Galla derivation produces the `correct' dCLE approximation of the underlying delay birth-death (dBD) process in the sense that the first and second moments of the dCLE match those of the dBD process.  However, their derivation has some limitations (see Section~\ref{sec:discuss}).  In particular, it gives no rigorous quantitative information about the distance between the dBD process and the dCLE.

In this paper, we establish a rigorous link between dBD processes and dCLEs by proving that the distance between the dBD process and the correct approximating dCLE process converges to zero as system size tends to infinity (as measured by expectations of functionals of the processes).  In particular, this result applies to all moments.  It is natural to express distance in terms of expectations of functionals because the dBD process is spatially discrete while the correct dCLE produces continuous trajectories (see Figure~\ref{fig:intro}).  Further, we prove that both processes converge weakly to the thermodynamic limit.  Finally, we quantitatively estimate the distance between the dBD process and the correct dCLE approximation as well as the distance of each of these to the thermodynamic limit.  All of these results hold for both fixed delay and distributed delay (see Figure~\ref{fig:schematic}A). 

The correct dCLE approximation is distinguished within the class of Gaussian approximations of the dBD process by the fact that it matches both the first and second moments of the dBD process.  As we will see, it performs remarkably well at moderate system sizes in a number of dynamical settings: steady state dynamics, oscillatory dynamics, and metastable switches.  We will demonstrate via simulation and argue mathematically using characteristic functions that no other Gaussian process with appropriately scaled noise performs as well.  In the following, the term `dCLE' shall refer specifically to the dCLE derived by Brett and Galla and expressed by~\eqref{e:dcle_general}, unless specifically stated otherwise.
We prove our mathematical results in the supplement~\cite{supplementary_material}.

\section{Simulations/Outline and interpretation of results} 
\label{sec:results}
  
Genetic regulatory networks may be simulated using an exact dSSA to account for 
transcriptional delay~\cite{Bratsun_Volfson_Tsimring_Hasty_2005, Barrio_Burrage_Leier_Tian_2006, Josic_Lopez_Ott_Shiau_Bennett_2011, Schlicht_Winkler_2008}.  Here we provide a heuristic derivation of a related
 dCLE, and show that in a number of concrete examples it provides an excellent approximation of the system (see Figure~\ref{fig:schematic}).  These simulations raise the following questions: Is the dCLE approximation valid in general? Can the expected quality of the approximation be quantified in general?  We answer these questions mathematically in Section~\ref{sec:main}.

We will adopt the following notation for reactions with delay,
\[
X + Y \xdashrightarrow[\mu]{\alpha(X, Y)} Z
\] 
Here $\alpha (X,Y)$ denotes the rate of the reaction, the dashed arrow indicates a reaction with delay, and $\mu $ is a probability measure that describes the delay distribution.  Solid arrows indicate reactions without delay.

\begin{figure}
\begin{center}
\includegraphics[scale=1.0]{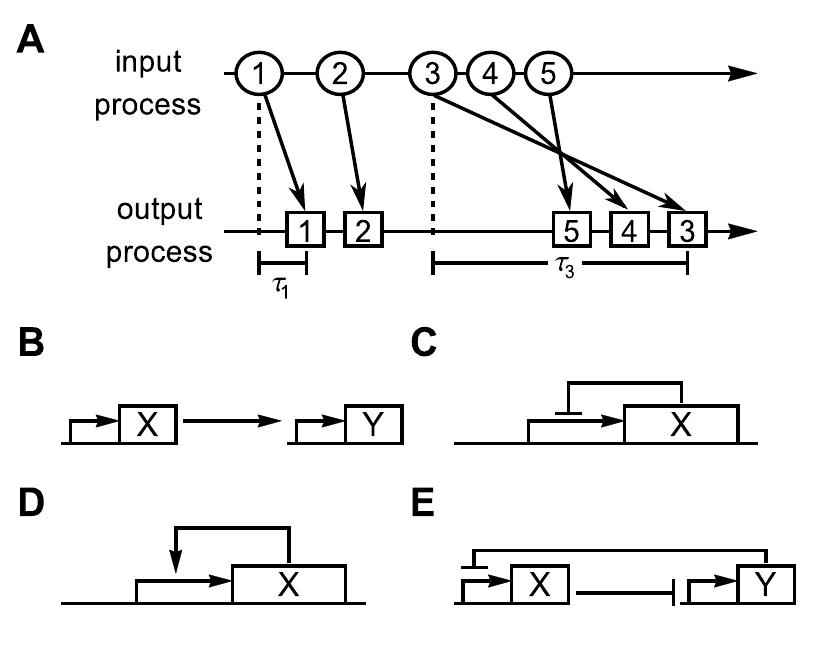}
\caption{(A) The effect of distributed delay on the protein production process. The ``input process'' is the first step in the transcription process, while the ``output process'' is the final mature product that enters the population.  The time delay $\tau $ accounts for the lag between the initialization of transcription and the production of mature product.  In a system with distributed delay, different production events can have different delay times; the order of the output process may therefore not match that of the input process.  (B--E) Simulated gene regulatory network motifs: a transcriptional  cascade  (B), oscillators  (C) and metastable systems (D--E).
}
\label{fig:schematic}
\end{center}
\end{figure}

\subsection{A transcriptional cascade}

First we consider a transcriptional cascade  with  two genes  that code for proteins $X$ and $Y$.  Protein $X$ is produced at a basal rate; production of $Y$ is induced by the presence of $X$.  The state of the system is represented by an ordered pair $(X,Y)$.  Note that we use $X$ and $Y$ to denote both protein names and protein numbers.  The reactions in the network, and the associated state change vectors $v_{i}$, are given by
\begin{align}
\emptyset \xdashrightarrow[\mu]{a} X  \quad\quad\quad &v_1 = (1,0)\\
X\xrightarrow{b_1 X} \emptyset  \quad\quad\quad &v_2 = (-1,0)\\
\emptyset \xdashrightarrow[\mu]{\psi (X)} Y  \quad\quad\quad &v_3 = (0,1)\label{reac:type3}\\
Y\xrightarrow{b_2 Y} \emptyset  \quad\quad\quad &v_4 = (0,-1)
\end{align}

This system can be simulated exactly using the dSSA: Suppose the state of the system and the reactions in the queue are known at time $t_0$ (the queued reactions can be thought of as the ``input process''; see Figure~\ref{fig:schematic}A), and that the delay kernel $\mu$ is supported on a finite interval $[0,\tau_0]$~\cite{dSSA_note}.

\begin{enumerate}[leftmargin=*, label=(\arabic*), ref=\arabic*]
\item Sample a waiting time $t_{w}$  from an exponential distribution with parameter $A_0 \defas a + b_1 X + \psi (X) + b_2 Y$.
\item If there is a reaction in the queue that is scheduled to exit at time $t_{q} < t_0 + t_{w}$,  advance to time $t_q$ and set $t_0 \rightarrow t_q$ and $(X, Y) \mapsto (X,Y) + (v_i^{(1)}, v_i^{(2)})$, where $v_i = (v_i^{(1)}, v_i^{(2)})$ is the change in the system due to the scheduled reaction.  Finally, sample a new waiting time for the next reaction.  
\item If no reaction exits before $t_0 + t_{w}$, set $t_0 \mapsto t_0 + t_{w}$ and sample a reaction type from the set $\left\{ 1,2,3,4\right\}$ with probabilities proportional to $\left\{ a, b_1 X, \psi (X), b_2 Y\right\}$, respectively.  If the reaction chosen is a non-delayed reaction, perform the update $(X, Y) \mapsto (X-1, Y)$ (or $(X, Y) \mapsto (X, Y-1)$). However, if the reaction chosen is a delayed reaction, the state change vector $(1,0)$ (or $(0,1)$) is put into the queue along with an exit time $t_q$. The difference $\tau = t_q - t_0$ between the current  and the exit time is sampled from the delay distribution $\mu$.
\end{enumerate}

We now heuristically derive the dCLE for the feed-forward system from this spatially discrete process.

Suppose the delay kernel $\mu $ is given by a probability density function $\kappa $ supported on $[0,\tau_{0}]$ ($\mrm{d} \mu (s) = \kappa (s) \mrm{d} s$).  We first approximate the number of reactions that produce $Y$ (Eq.\ (\ref{reac:type3})) that will be completed within the interval $[t_{0}, t_{0} + \Delta ]$, where $t_{0}$ denotes the current time and $\Delta $ is a small increment.  Since the production of $Y$ involves delay, a reaction of this type that is completed within $[t_{0}, t_{0} + \Delta ]$ must have been initiated at some time within $[t_{0} - \tau_{0}, t_{0}]$.  Let $t_{0} > t_{0} - \Delta > t_{0} - 2 \Delta > \cdots $ be the partition of $[t_{0} - \tau_{0}, t_{0}]$ into intervals of length $\Delta $.  The (random) number of reactions completed within $[t_{0}, t_{0} + \Delta ]$ {\itshape and} initiated within $[t_{0} - (i+1) \Delta , t_{0} - i \Delta ]$ may be approximated by a Poisson random variable with mean
\begin{equation*}
\psi (X_{t_{0} - (i+1) \Delta}) \Delta \cdot \kappa ((i+1) \Delta ) \Delta .
\end{equation*}
Summing over $i$, the (random) number of reactions completed within $[t_{0}, t_{0} + \Delta ]$ may be approximated by a Poisson random variable with mean
\begin{equation*}
\Delta \sum_{i} \left[ \psi (X_{t_{0} - (i+1) \Delta}) (\kappa ((i+1) \Delta ) \Delta ) \right];
\end{equation*}
this is a Riemann sum that approximates the integral
\begin{equation*}
\Delta \int_{0}^{\tau_{0}} \psi (X_{t_{0} - s}) \kappa (s) \, \mrm{d} s.
\end{equation*}
Known as $\tau $-leaping, this line of reasoning produces a Poissonian approximation of the dBD process:
\begin{align*}
\delta X_t &= \Poisson (a\Delta) - \Poisson(b_1 X_t \Delta )\\
\delta Y_t &= \Poisson \left(\Delta\int_{0}^{\tau_0} \psi (X_{t -s}) \, \mrm{d} \mu(s)\right) - \Poisson(b_2 Y_t \Delta ).
\end{align*}
Here $\Poisson (\eta )$ denotes a Poisson random variable with mean $\eta $.

\begin{figure}
\begin{center}
\includegraphics[scale=1.0]{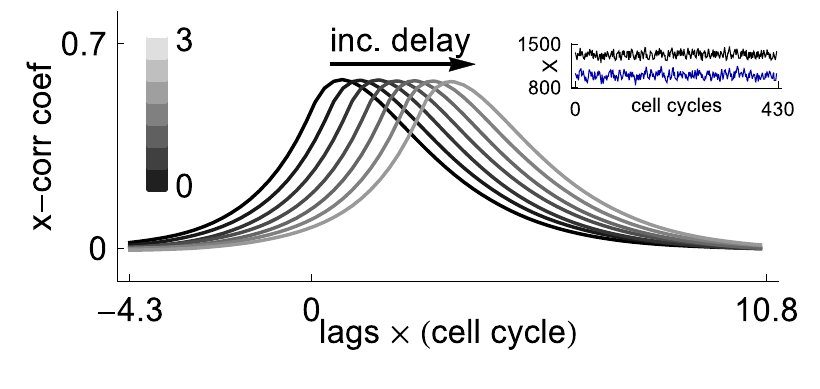}
\caption{Cross correlation functions for the two-species feed-forward architecture at system size $N = 1000$. The inset shows sample trajectories for $X$ (black; top) and $Y$ (blue; bottom). The mean field model has a fixed point, and for this reason, stochastic dynamics stay within a neighborhood of this fixed point (see Theorem \ref{thm:tube}). However, the effect of increasing delay is clearly seen in the cross correlation function. The color bar displays the delay size corresponding to the cross correlation curves.  Parameter values are given by $\tilde{\psi}(X) = \tilde{a}_1 x^3/(m^3+x^3)$, $\tilde{a} = 0.92$, $\tilde{a}_1 = 1.39$, $b_1 = b_2 = \ln(2)$, $m = 1.33$, $\mu = \delta_{\tau}$.  The cross correlations have been normalized by dividing by the standard deviations $\sigma_{X}$ and $\sigma_{Y}$ of $X$ and $Y$.  Time has been normalized to cell cycle length.}
\label{fig:ff_cross}
\end{center}
\end{figure}

If these Poisson random variables have large mean, they can be approximated by normal random variables. For example, the Poisson variable representing the number of reactions that produce $Y$ can be approximated by a normal random variable with mean and variance equal to $\Delta\int_{0}^{\tau_0} \psi (X_{t -s}) \, \mrm{d} \mu(s)$. Since each reaction changes the state of either $X$ or $Y$ (but never both), it follows that the evolution of the system can be approximated by the stochastic difference equation
\begin{subequations}
\label{e:s-difference-1}
\begin{align}
\delta X_t &= \Delta  (a-b_1 X_t)+ \sqrt{\Delta (a + b_1 X_t)} \mscr{N} (0,1) \\
\delta Y_t &=  \Delta\left(\int_{0}^{\tau_0} \psi (X_{t - s}) \, \mrm{d} \mu(s) - b_2 Y_t\right) \\
 &\quad {}+ \sqrt{\Delta\left(\int_{0}^{\tau_0} \psi (X_{t - s}) \, \mrm{d} \mu(s) + b_2 Y_t\right)} \mscr{N} (0,1),
\notag
\end{align} 
\end{subequations}
where $\mscr{N} (0,1)$ is the standard normal random variable.

System~\eqref{e:s-difference-1} may be written in terms of concentrations.  Let $N$ be a system size parameter.  We think of $N$ as a characteristic protein number; $X_{t} / N$ and $Y_{t} / N$ therefore represent fractions of this characteristic value.  Writing $x_{t} = X_{t} / N$, $y_{t} = Y_{t} / N$, $\tilde{\psi} (x) = \psi (Nx) / N$, and assuming that the basal production rate $a$ scales with $N$ as $a = \tilde{a} N$, we obtain
\begin{subequations}
\label{e:s-difference-2}
\begin{align} 
\delta x_t &= \Delta  (\tilde a- b_1 x_t) + \frac{1}{\sqrt{N}}\sqrt{\Delta (\tilde a + b_1 x_t)} \mscr{N} (0,1) \\
\delta y_t & =  \Delta\left(\int_{0}^{\tau_0} \tilde{\psi} (x_{t - s}) \, \mrm{d} \mu(s) - b_2 y_t\right) \\
 &\quad {}+ \frac{1}{\sqrt{N}}\sqrt{\Delta\left(\int_{0}^{\tau_0} \tilde{\psi} (x_{t - s}) \, \mrm{d} \mu(s) + b_2 y_t\right)} \mscr{N} (0,1).
\notag
\end{align}
\end{subequations}

Eq.~\eqref{e:s-difference-2} is the Euler--Maruyama type discretization of a delay stochastic differential equation.  Replacing $\Delta $ with $\mrm{d} t$ and $\sqrt{\Delta} \mscr{N} (0,1)$ with $\mrm{d} W_{t}$ in~\eqref{e:s-difference-2}, we obtain
\begin{subequations}
\label{e:ff-sde}
\begin{align} 
\mrm{d} x_t &= (\tilde a- b_1 x_t) \, \mrm{d} t + \frac{1}{\sqrt{N}}\sqrt{(\tilde a + b_1 x_t)} \, \mrm{d} W^{1}_{t} \\
\mrm{d} y_t & = \left(\int_{0}^{\tau_0} \tilde{\psi} (x_{t - s}) \, \mrm{d} \mu(s) - b_2 y_t \right) \mrm{d} t \\
 &\quad {}+ \frac{1}{\sqrt{N}}\sqrt{\left(\int_{0}^{\tau_0} \tilde{\psi} (x_{t - s}) \, \mrm{d} \mu(s) + b_2 y_t\right)} \, \mrm{d} W^{2}_{t}.
\notag
\end{align}
\end{subequations}
This is the dCLE for the transcriptional cascade in this section.

Taking the formal  thermodynamic limit, $N \to \infty $, in Eq.~\eqref{e:ff-sde} yields the reaction rate equations derived in~\cite{Schlicht_Winkler_2008}:
\begin{subequations}
\label{e:ff-dde}
\begin{align}
\mathrm{d} x_{t} &=  \left(\tilde a  - b_{1} x_{t} \, \right)\mathrm{d} t \\
\mathrm{d} y_{t} &=\left( \int_{0}^{\tau_{0}} \tilde{\psi} (x_{t -s}) \, \mathrm{d} \mu (s) - b_{2} y_{t} \right) \mathrm{d} t.
\end{align}
\end{subequations}
The dynamics described by Eq.~\eqref{e:ff-dde} are quite simple; if $b_{1} > 0$ and $b_{2} > 0$, then~\eqref{e:ff-dde} has a globally attracting stable stationary point. 

To test the validity of the dCLE approximation~\eqref{e:ff-sde}, we examine if it captures the interaction between the two proteins in our transcriptional cascade network.  Figure~\ref{fig:ff_cross} shows the cross correlation functions obtained by simulating the system with $N = 1000$ using dSSA.  From left to right, the curves correspond to fixed delay $\tau $ increasing from $0$ to $3$.  The corresponding cross correlation curves for the dCLE approximation~\eqref{e:ff-sde} are indistinguishable from those obtained using dSSA.

In the heuristic derivation above, we first fix $N$ and let $\Delta \to 0$ to obtain the dCLE; we then separately let $N \to \infty $ to obtain the thermodynamic limit.  Brett and Galla~\cite{Brett_Galla_2013} also derive the dCLE by first fixing $N$ and then sending $\Delta \to 0$.  However, the two limits, $\Delta \to 0$ and $N \to \infty $, cannot be taken independently; this is a common problem with heuristic derivations of stochastic differential equations, even in the absence of delay~\cite{ethier1986markov}. The time discretization, $\Delta,$ can be thought of as a sampling frequency, while the system size, $N,$ determines the rate at which reactions fire. If $N$ becomes too large for a given $\Delta$, then the number of reactions that fire within $[t,t + \Delta ]$ no longer follows a Poisson distribution with mean dependent only on the state of the system at time $t$. On the other hand, if $N$ is too small for a given $\Delta $, then the Poisson distribution cannot be approximated by a normal distribution.  In order to rigorously derive the Langevin approximation  and estimate the distance between the dBD and dCLE processes, we will have to take a careful limit by relating $\Delta$ to $N$ (with $\Delta \to 0$ as $N \to \infty $).  We describe the proper scaling in Section~\ref{sec:main}.

Applied to the transcriptional cascade, Theorem~\ref{thm:dssa_dcle} asserts that provided $\Delta $ scales correctly with $N$, the distance between the dBD process and the process described by Eq.~\eqref{e:ff-sde} converges to zero as $N \to \infty $ (as measured by expectations of functionals of the processes).  Theorem~\ref{thm:dssa_therm} asserts that the dBD process then converges weakly to the thermodynamic limit given by Eq.~\eqref{e:ff-dde} as $N \to \infty $.  Moreover, when $\Delta $ is correctly scaled with respect to $N$, Theorem~\ref{thm:tube} provides explicit bounds for the probabilities that the dBD and dCLE processes deviate from a narrow tube around the solution of Eq.~\eqref{e:ff-dde}.

In the previous example, the deterministic system has a fixed point. The time series for the stochastic system, therefore, stay within a small neighborhood of this fixed point (see inset, Fig.~\ref{fig:ff_cross}). In the next example, we show that the dCLE approximation remains excellent even when the deterministic dynamics are non-trivial. We consider a degrade-and-fire oscillator for which the deterministic system has a limit cycle. The dCLE correctly captures the peak height and the inter-peak times for the dSSA realization of the degrade and fire oscillator, in addition to statistics such as the mean and variance. The approximation does not break down at small instantaneous protein numbers. Indeed, the mathematical theory developed in this work makes an important point: protein concentrations at any particular time do not limit the quality of the dCLE approximation (in the presence of delay, or otherwise).  Instead, the quality of the dCLE approximation depends on the latent parameter $N$. Theorem~\ref{thm:tube} makes this more precise: if one fixes the allowable error $\varepsilon$ in the approximation of the dBD process by the dCLE process, then the time $T$ during which the approximation error stays smaller than $\varepsilon$ increases with $N$.

\begin{figure}
\begin{center}
\includegraphics[scale=1.0]{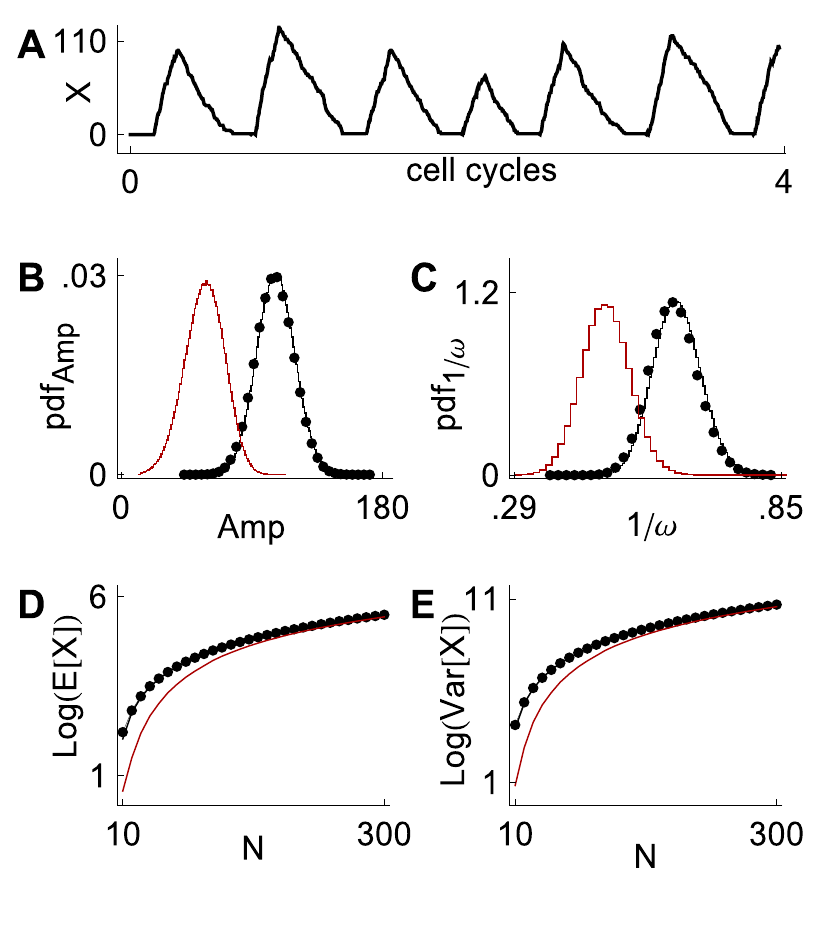}
\caption{Comparison of dSSA results to dSDE approximations for the degrade and fire oscillator.  (A) depicts a stochastic realization of the oscillator generated by dSSA.  We compare dSSA statistics (black dots in (B)--(E)) to those generated by the dCLE approximation given by Eq.~\eqref{e:df_dcle_correct} (black curves).  We also show results for the dSDE approximation given by Eq.~\eqref{e:df_dcle_incorrect}, obtained by removing delay from the diffusion term in Eq.~\eqref{e:df_dcle_correct} (red curves).  At system size $N = 50$, the dCLE approximation given by Eq.~\eqref{e:df_dcle_correct} closely matches dSSA with respect to spike height distribution (B) and interspike interval distribution (C).  In contrast, removing delay from the diffusion term results in a poor approximation of these distributions, as shown by the sizable shifts affecting the red curves.  (D) and (E) illustrate mean repressor protein level and repressor protein variance, respectively, as functions of system size $N$.  Eq.~\eqref{e:df_dcle_correct} provides a good approximation for all simulated values of $N$ while the performance of Eq.~\eqref{e:df_dcle_incorrect} improves as $N$ increases.  The quantity $P$ represents protein number, not protein concentration.  Parameter values are $\alpha = 20.8$, $C_1 = 0.04$, $\beta = \ln(2)$, $V_{max} = 5.55$, $\gamma_0 = 0.01$, $\mu = \delta_{0.14}$.  A soft boundary was added at $0$ to ensure positivity.}
\label{fig:df_oscillator}
\end{center}
\end{figure}

\subsection{Degrade and fire oscillator}\label{sec:dnf}

The degrade and fire oscillator depicted schematically in Figure~\ref{fig:schematic}C consists of a single autorepressive gene and corresponds to the reaction network 
\begin{align*}
&\emptyset \xdashrightarrow[\mu]{\psi (X)} X \xrightarrow{\gamma X} \emptyset\\
&X + E \xrightarrow{\eta (X)} E
\end{align*}
The production rate $\psi (X)$ is given by $\psi (X) = N f (X/N)$, where $f$ is the propensity function 
\[
f(x) = \frac{\alpha }{1 + \left( x / C_{1} \right)^{4}};
\]
the enzymatic degradation rate $\eta (X)$ is given by $\eta (X) = N g(X/N)$. Here
$g(x) = V_{max} x / (K + x)$;
$K$ is the Michaelis-Menten constant, $V_{max}$ the maximal enzymatic degradation rate, and $\gamma$ the dilution rate coefficient.  In the thermodynamic limit, the system is modeled by the delay differential equation
\begin{equation}
\label{e:df_dde}
\frac{ \mrm{d} x }{ \mrm{d} t } = \int_{0}^{\tau_{0} } \frac{ \alpha }{ 1 + \left( \frac{ x(t-s) }{ C_{1} } \right)^{4} } \, \mrm{d} \mu (s) - \gamma x(t) - \frac{ V_{max} x(t) }{ K + x(t) }.
\end{equation}
As before, $x(t)$ denotes the concentration of  protein $X$.  We model the formation of functional repressor protein using distributed delay (described by the probability measure $\mu $); this delayed negative feedback can induce oscillations~\cite{Mather_Bennett_Hasty_Tsimring_2009}.  Figure~\ref{fig:df_oscillator}A depicts a sample realization of the stochastic version of the degrade and fire oscillator (the finite system size regime) generated by dSSA.

The dCLE approximation is in this case given by
\begin{equation}
\label{e:df_dcle_correct}
\begin{aligned}
\mrm{d} &x_{t} = \int_{0}^{\tau_{0} } f(x_{t-s}) \, \mrm{d} \mu (s) - \gamma x_{t} - g(x_{t}) \, \mrm{d} t \\
&{}+ \frac{1}{\sqrt{N}} \left[ \int_{0}^{\tau_{0} } f(x_{t-s}) \, \mrm{d} \mu (s) + \gamma x_{t} + g(x_{t}) \right]^{\frac{1}{2}} \mrm{d} W_{t}.
\end{aligned} 
\end{equation}
Figure~\ref{fig:df_oscillator} illustrates that Eq.~\eqref{e:df_dcle_correct} provides a good approximation of the dBD dynamics, even when system size is relatively small.  At system size $N = 50$, the spike height distribution and interspike interval distribution obtained using the dSSA (black dots in Figure~\ref{fig:df_oscillator}B--\ref{fig:df_oscillator}C) are nearly indistinguishable from those obtained using Eq.~\eqref{e:df_dcle_correct} (black curves in Figure~\ref{fig:df_oscillator}B--\ref{fig:df_oscillator}C).  Further, we see a close match with respect to mean repressor protein level and repressor protein variance across a range of system sizes (Figure~\ref{fig:df_oscillator}D--\ref{fig:df_oscillator}E).

Interestingly, the dCLE approximation is very good even though the protein number approaches zero during part of the oscillation.  This illustrates a central feature of the theory: the quality of the dCLE approximation is a function of a latent parameter $N$, not of the number of molecules present at any given time. 

The exact form of the diffusion term is crucial to the accuracy of the dCLE approximation.  If we remove delay from the diffusion term in Eq.~\eqref{e:df_dcle_correct}, we obtain
\begin{equation} 
\label{e:df_dcle_incorrect}
\begin{aligned}
\mrm{d} x_{t} &= \int_{0}^{\tau_{0} } f(x_{t-s}) \, \mrm{d} \mu (s) - \gamma x_{t} - g(x_{t}) \, \mrm{d} t \\
&\quad {}+ \frac{1}{\sqrt{N}} \left[ f(x_{t}) + \gamma x_{t} + g(x_{t}) \right]^{1/2} \mrm{d} W_{t}.
\end{aligned} 
\end{equation}
At system size $N = 50$ (red curves in Figure~\ref{fig:df_oscillator}B--\ref{fig:df_oscillator}C), dSDE~\eqref{e:df_dcle_incorrect} produces dramatically different results from those generated by the correct dCLE approximation.  The performance of Eq.~\eqref{e:df_dcle_incorrect} improves as $N$ increases (Figure~\ref{fig:df_oscillator}D--\ref{fig:df_oscillator}E).  This is expected, as both Eq.~\eqref{e:df_dcle_correct} and Eq.~\eqref{e:df_dcle_incorrect}
converge weakly to Eq.~\eqref{e:df_dde} as $N \to \infty$.

\subsection{Metastable systems}

Understanding metastability in stochastic systems is of fundamental importance in the study of biological switches~\cite{XiaoWang_PNAS_2013, hong:2012, he:2011, Ozbudak2004, Kepler_Elston_2001, aurell:2002, warren:2005, balaban:2004, gardner:2000, nevozhay:2012}.  While metastability is well understood mathematically in the absence of delay, understanding the impact of delay on metastability remains a major theoretical and computational challenge~\cite{Gupta_Lopez_Ott_Josic_Bennett_2013, Masoller_2003, Tsimring_Pikovsky_2001}.  We examine two
canonical examples to show that the dCLE can be used to study the impact of delay on metastability: a positive feedback circuit and a co-repressive genetic toggle switch.

\subsubsection{Single species positive feedback circuit}\label{sec:pf}

\begin{figure}[t]
\begin{center}
\includegraphics[scale=1.0]{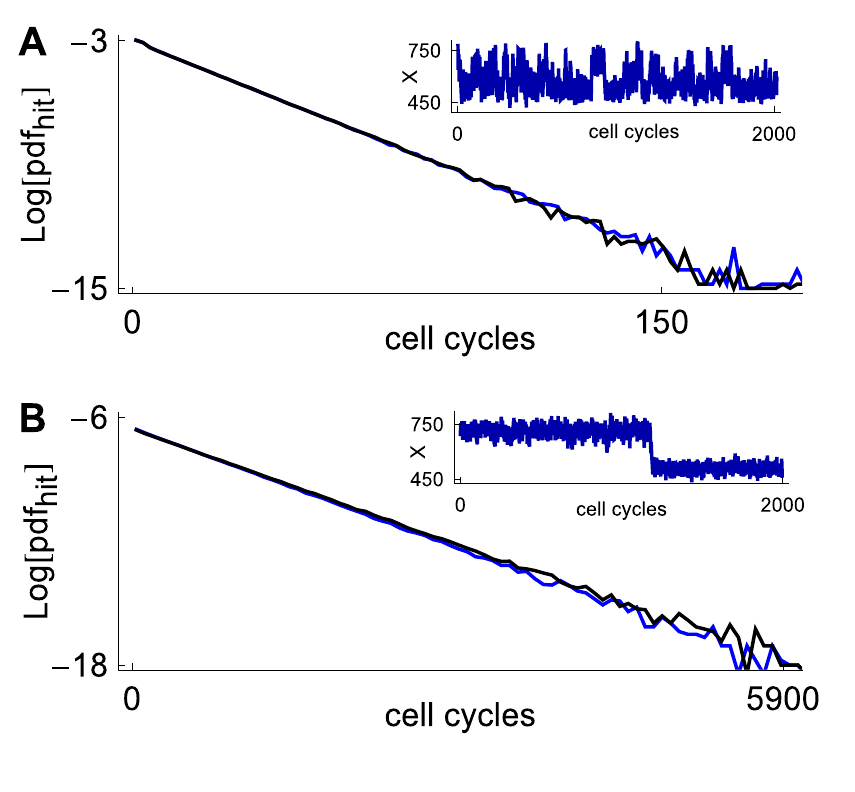}
\caption{Hitting time distributions for the positive feedback circuit.  Black curves represent dSSA data; blue curves represent data from the dCLE approximation.  The top panel corresponds to the Markov case $\tau = 0$ and a Hill coefficient of $b = 15$. The bottom row corresponds to a delay $\tau = 0.75$ and $b = 25$. The tail of the hitting times distribution becomes longer with both increasing delay and increasing Hill coefficient $b$. The dCLE captures the lengthening of the tail due to both effects.
Parameter values are $\alpha = 0.35$, $\beta = 0.15$, $c = 0.615$, $\gamma = \ln(2)$, $\mu = \delta_{\tau}$, $N = 1000$.}
\label{fig:pf_hitting}
\end{center}
\end{figure}

The simplest metastable system consists of a single protein that drives its own production (Figure~\ref{fig:schematic}D).  The chemical reaction network is given by
\[
\emptyset \xdashrightarrow[\mu]{\psi (X)} X \xrightarrow{\gamma X} \emptyset
\] 
with $\psi (X) = N f(X/N)$ for the propensity 
\begin{equation*}
f(x) = \alpha + \frac{\beta x^b }{c^b + x^b}.
\end{equation*}  
In the thermodynamic limit,  the dynamics of this model are described by the DDE
\begin{equation}
\label{e:pf_dde}
\mathrm{d} x_{t} = \int_{0}^{\tau_{0}} \alpha + \beta \frac{x(t - s)^{b}}{c^{b} + x(t-s)^{b}} \, \mathrm{d} \mu (s) - \gamma x(t) \, \mathrm{d} t.
\end{equation}
Here $x$ represents protein concentration and $b$ is the Hill coefficient.  In the thermodynamic limit, there are two stable stationary states, $x_{l}$ and $x_{h}$, as well as an unstable stationary state $x_{s}$.  These states satisfy $x_{l} < x_{s} < x_{h}$.

In the stochastic (finite $N$) regime, the stationary states $x_{l}$ and $x_{h}$ become metastable.  We simulate the metastable dynamics using dSSA and the dCLE approximation~\eqref{e:dcle_general} given in this case by
\begin{equation}
\label{e:pf_dcle}
\begin{aligned}
\mathrm{d} x_{t} &= \int_{0}^{\tau_{0}} f(x_{t-s}) \, \mathrm{d} \mu (s) - \gamma x_{t} \, \mathrm{d} t \\
&\quad {}+ \frac{1}{\sqrt{N}} \left[ \int_{0}^{\tau_{0}} f(x_{t-s}) \, \mathrm{d} \mu (s) + \gamma x_{t} \right]^{\frac{1}{2}} \mathrm{d} W_{t}.
\end{aligned}
\end{equation}
Figure~\ref{fig:pf_hitting} displays hitting time distributions for the dSSA simulations (black curves) and Eq.~\eqref{e:pf_dcle} (blue curves).  A hitting time is defined as follows:  We choose neighborhoods $(x_{l} - \delta_{l}, x_{l} + \delta_{l})$ and $(x_{s} - \delta_{s}, x_{s} + \delta_{s})$ of $x_{l}$ and $x_{s}$, respectively.  We start the clock when a trajectory enters $(x_{l} - \delta_{l}, x_{l} + \delta_{l})$ from the right.  The clock is stopped when that trajectory first enters $(x_{s} - \delta_{s}, x_{s} + \delta_{s})$.  A hitting time is the amount of time that elapses from clock start to clock stop. 
We see that for no delay (Figure~\ref{fig:pf_hitting}, top) and fixed delay $\tau = 1$ (bottom), the dCLE approximation accurately captures the hitting time distributions for Hill coefficients increasing from $15$ to $25$.   
Hence, the dCLE approximation accurately captures the rare events associated with a spatially-discrete delay stochastic process. This is significant because dSDEs are more amenable to large deviations theoretical analysis than their spatially-discrete counterparts.

Hitting times increase dramatically as the delay increases from $0$ to $1$, in accord with earlier analysis~\cite{Gupta_Lopez_Ott_Josic_Bennett_2013}.  A dramatic increase is also seen as the Hill coefficient increases.  This is due to the fact that the potential wells around $x_{l}$ and $x_{h}$ deepen as $b$ increases.

\subsubsection{Co-repressive toggle switch}

\begin{figure}[ht]
\begin{center}
\includegraphics[scale=0.75]{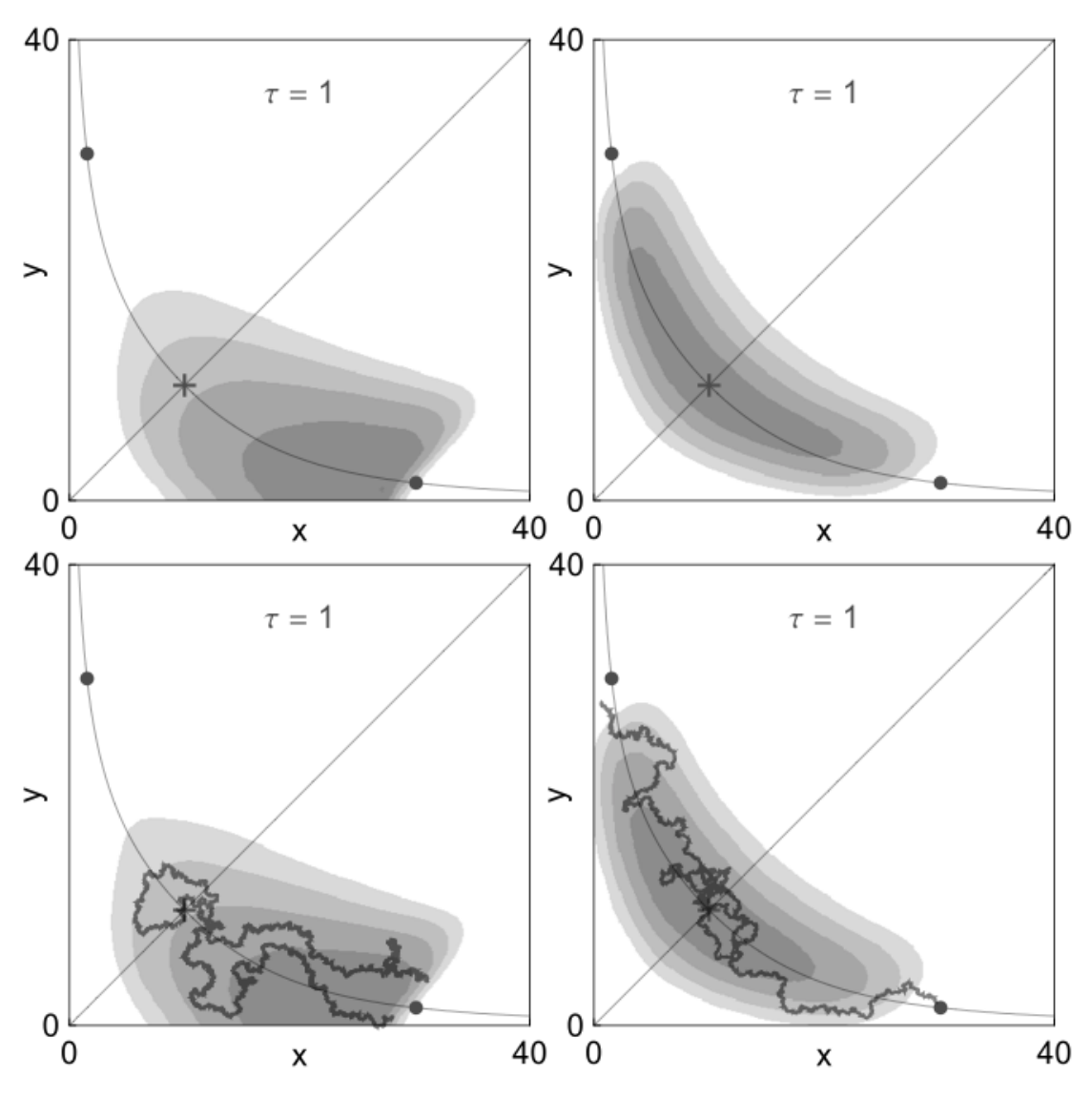}
\caption{Density plots for the co-repressive toggle switch.  The two left panels illustrate trajectories that leave a neighborhood of the stable point $(x_{h},y_{l})$ and fall back into the same neighborhood before transitioning into a neighborhood of the stable point $(x_{l},y_{h})$.  The two right panels illustrate trajectories that leave a neighborhood of $(x_{h},y_{l})$ and transition to a neighborhood of $(x_{l},y_{h})$ before falling back into the first neighborhood. Top panels correspond to dSSA; bottom panels correspond to dCLE. Cartoons of typical trajectories corresponding to failed transitions (left) and successful transitions (right) are shown for the dCLE process.  Plots are shown for $N = 30$. The values of the other parameters are $\beta = 0.73$, $k = 0.05$, $\gamma = \ln(2)$.}
\label{fig:toggle}
\end{center}
\end{figure}

The co-repressive toggle switch (Figure~\ref{fig:schematic}E) is a two-dimensional metastable system described in the thermodynamic limit by the DDEs
\begin{subequations}
\label{e:toggle_dde}
\begin{align}
\mathrm{d} x_{t} &= \int_{0}^{\tau_{0}} \frac{\beta }{1 + y(t-s)^{2}/k} \, \mathrm{d} \mu_{1} (s) - \gamma x \, \mathrm{d} t \\
\mathrm{d} y_{t} &= \int_{0}^{\tau_{0}} \frac{\beta }{1 + x(t-s)^{2}/k} \, \mathrm{d} \mu_{1} (s) - \gamma y \, \mathrm{d} t.
\end{align}
\end{subequations}
The measure $\mu_{1}$ describes the delay associated with production in this symmetric circuit.  Eq.~\eqref{e:toggle_dde} has two stable stationary points $(x_{l}, y_{h})$ and $(x_{h}, y_{l})$ separated by the unstable manifold associated with a saddle equilibrium point $(x_{s}, y_{s})$. In the stochastic (finite system size) regime, the stable stationary points become metastable.  In this regime a typical trajectory  spends most of its time near the metastable points, occasionally moving between them.

Figure~\ref{fig:toggle} displays density plots corresponding to trajectories that either successfully transition between metastable states (four panels on the right) or make failed transition attempts (four panels on the left).  Even for the moderate system size $N = 30$, the density plots generated by dSSA (top four panels) closely match those generated by the dCLE approximation in Eq.~\eqref{e:dcle_general} (bottom four panels).

Given the importance of rare events throughout stochastic dynamics, it is encouraging that the dCLE approximation  captures their statistics well.

\section{Main Results}\label{sec:main}

The simulations thus far described suggest that the dCLE closely approximates the dBD process provided that $\Delta $ scales properly with $N$.  We next provide mathematical statements that make this observation precise.  We prove that the distance between the dBD process and the approximating dCLE (as measured by expectations of functionals of the processes) converges to zero as the system size $N \to \infty $ (Theorem~\ref{thm:dssa_dcle}).  In particular, Theorem~\ref{thm:dssa_dcle} implies that the dCLE may be used to approximate all moments of the dBD process.  Further, we then prove that the dBD and dCLE processes both converge weakly to the thermodynamic limit (Theorem~\ref{thm:dssa_therm}).  Theorem~\ref{thm:dssa_therm} strengthens a result of Schlicht and Winkler~\cite{Schlicht_Winkler_2008} that establishes convergence of the first moment of the dBD process.  Theorem~\ref{thm:tube} quantitatively bounds the probabilities that the dBD and dCLE processes deviate from a narrow tube around the solution of the deterministic thermodynamic limit.

We first precisely describe the general setting and then state our theorems.  All proofs are provided in the supplement~\cite{supplementary_material}.

Consider a system of $D$ biochemical species and $M$ possible reactions.  We are interested in describing the dynamics as a function of a latent system parameter $N$, the system size.  Let $B_{N} (t) \in \mbb{Z}^{D}$ denote the state of the system at time $t$.

Each reaction $R_{j}$ is described by the following:
\begin{enumerate}[leftmargin=*, label=(\alph*), ref=\alph*]

\item
A propensity function $f_{j} : \mbb{R}^{D} \to \mbb{R}^{+}$.  The firing rate of reaction $R_{j}$ is given by $N f_{j} (B_{N} (t) / N)$.

\item
A state-change vector $\mbi{v}_{j} \in \mbb{Z}^{D}$.  The vector $\mbi{v}_{j}$ describes the change in the number of molecules of each species that results from the completion of a reaction of type $j$.

\item
A probability measure $\mu_{j}$ supported on $[0, \tau_{0}]$.  The measure $\mu_{j}$ models the delay that may occur between the initiation and completion of a reaction of type $j$.  If reaction $R_{j}$ is instantaneous, then $\mu_{j} = \delta_{0}$.  If the delay is a fixed value $\tau > 0$, then $\mu_{j} = \delta_{\tau }$.  If $\mu_{j}$ has a density, we denote it by $\kappa_{j}$.
\end{enumerate}

Given a system trajectory up to time $t$, $\vset{ B_{N} (s) : s \leqs t }$, the dSSA can be described as follows:

\begin{enumerate}[leftmargin=*, label=(I\arabic*), ref=I\arabic*]

\item
\label{i:time_reaction}
Sample a time $\xi $ to the next reaction from the exponential distribution with rate
\begin{equation*}
N \sum_{j=1}^{M} f_{j} \left( \frac{ B_{N} (t) }{ N } \right).
\end{equation*}

\item
Select a reaction $R_{k}$  with probability
\begin{equation*}
\frac{ f_{k} (B_{N} (t) / N) }{ \sum_{j=1}^{M} f_{j} (B_{N} (t) / N) }.
\end{equation*}

\item
Sample a delay time $\tau $ from $\mu_{k}$.

\item
If no delayed reactions are set to finish in the time interval $[t, t + \xi ]$, then proceed as follows.  If $\tau = 0$, then move to time $t + \xi $ and set $B(t + \xi ) = B(t) + \mbi{v}_{k}$.  If $\tau > 0$, then move to time $t + \xi $, set $B(t + \xi ) = B(t)$, and put the state change vector $\mbi{v}_{k}$ into a queue along with the designated time of reaction completion, $t + \xi + \tau $.

\item
If a reaction from the past is set to complete in $[t, t + \xi ]$, then move to the time $\hat{t}$ of the completion of the first such reaction, update $B(\hat{t})$ accordingly, and then proceed to~\pref{i:time_reaction}.

\end{enumerate}

Schlicht and Winkler~\cite{Schlicht_Winkler_2008} prove the existence of the stochastic process $(B_{N} (t))$ from which the dSSA samples.  Let $\mbb{B}_{N}$ denote the probability measure on realizations $\omega_{B_{N}/N} (t)$ associated with the scaled process $(B_{N} (t) / N)$.  These realizations lie in the space $\mcal{D}$ of right-continuous functions from $[0,T]$ into $\mbb{R}^{D}$ that possess limits from the left.  Our main results quantify the behavior of $\mbb{B}_{N}$ as $N \to \infty $.  We make the following regularity assumptions on the propensities $f_{j}$.

\begin{enumerate}[leftmargin=*, label=(P\arabic*), ref=P\arabic*]

\item
\label{i:propen_reg}
The functions $f_{j}$ have continuous derivatives of order $2$.

\item
\label{i:propen_supp1}
There exists a compact set $K$ such that $\supp (f_{j}) \subset K$ for all $1 \leqs j \leqs M$.

\item
\label{i:propen_supp2}
For all $1 \leqs j \leqs M$, $f_{j} (\mbi{x}) > 0$ only if all coordinates of the vector $\mbi{v}_{j} + \mbi{x}$ are nonnegative; $f_{j} (\mbi{x}) = 0$ otherwise.

\end{enumerate}

The correct dCLE approximation of $(B_{N} (t) / N)$ is given for $1 \leqs k \leqs D$ by
\begin{equation}
\label{e:dcle_general}
\begin{aligned}
\mrm{d} x_{k} &= \left( \sum_{j=1}^{M} \int_{0}^{\tau_{0} } v_{jk} f_{j} (\mbi{x} (t - s)) \, \mrm{d} \mu_{j} (s) \right) \mrm{d} t \\
&\qquad \qquad {}+ \frac{1}{\sqrt{N}} (\Sigma \, \mathrm{d} \msf{W})_{k},
\end{aligned}
\end{equation}
where $\msf{W}$ is a $D$-dimensional vector of independent standard Brownian motions and $\Sigma^{2}$ is given by
\begin{equation*}
\Sigma^{2}_{lm} = \sum_{j=1}^{M} v_{jl} v_{jm} \int_{0}^{\tau_{0} } f_{j} (\mbi{x} (t-s)) \, \mathrm{d} \mu_{j} (s).
\end{equation*}

Let $L_{N}$ denote the stochastic process described by Eq.~\eqref{e:dcle_general} and let $\mbb{L}_{N}$ denote the probability measure on realizations $\omega_{L_{N}} (t)$ associated with this process.  Theorem~\ref{thm:dssa_dcle} controls the distance between $\mbb{B}_{N}$ and $\mbb{L}_{N}$.  

\begin{m_theorem}
\label{thm:dssa_dcle}
Assume that the propensities $f_{j}$ satisfy~\pref{i:propen_reg}--\pref{i:propen_supp2}.  Fix $T > 0$.  For every continuous observable $\Psi : \mcal{D} \to \mbb{R}$, we have
\begin{equation*}
\int_{\mcal{D}} \Psi (\omega_{B_{N}/N}) \, \mathrm{d} \mbb{B}_{N} - \int_{\mcal{D}} \Psi (\omega_{L_{N}}) \, \mathrm{d} \mbb{L}_{N} \to 0
\end{equation*}
as $N \to \infty $.
\end{m_theorem}

Theorem~\ref{thm:dssa_therm} establishes weak convergence of the scaled process $(B_{N} (t) / N)$ to the thermodynamic limit governed by the delay reaction rate equations
\begin{equation}
\label{e:drre_general}
\mrm{d} x_{k} = \left( \sum_{j=1}^{M} \int_{0}^{\tau_{0} } v_{jk} f_{j} (\mbi{x} (t - s)) \, \mrm{d} \mu_{j} (s) \right) \mrm{d} t.
\end{equation}

\begin{m_theorem}
\label{thm:dssa_therm}
Assume that the propensities $f_{j}$ satisfy~\pref{i:propen_reg}--\pref{i:propen_supp2}.  Fix $T > 0$.  For every continuous observable $\Psi : \mcal{D} \to \mbb{R}$, we have   
\begin{equation*}
\int_{\mcal{D}} \Psi (\omega_{B_{N}/N}) \, \mathrm{d} \mbb{B}_{N} \to \Psi (\mbi{x})
\end{equation*}
as $N \to \infty $, where $\mbi{x}$ denotes the solution of Eq.~\eqref{e:drre_general}.
\end{m_theorem}

Crucial to the proofs of Theorems~\ref{thm:dssa_dcle} and~\ref{thm:dssa_therm} are the discretization of the time interval $[0,T]$ and the development of quantitative controls on processes that approximate $b_{N} \defas B_{N} / N$.  We partition $[0,T]$ into subintervals of length $\Delta = \Delta (N)$.  It is crucial that $\Delta (N)$ scales with $N$ in the right way.  The quantitative pathwise controls in Theorem~\ref{thm:tube} hold if $\Delta (N) = N^{-1/4}$.

We define an Euler-Maruyama discretization $x_{N}$ of~\eqref{e:dcle_general}.

\begin{definition}[The process $x_N(t)$] 

For $t \leqs 0$, define $x_N(t) = b_N(t)$. For integers $k \geqs 1$ and $t = k \Delta $, define $x_N(t)$ recursively by
\begin{equation*}
x_N (k\Delta) = x_N((k-1)\Delta) + A_{1} + A_{2},
\end{equation*}
where
\begin{align*}
A_{1} &= \left[ \sum_{j = 1}^{M} v_j \int_{0}^{\infty} f_j(x_N((k-1)\Delta -s))~d\mu_j(s) \right] \Delta ,
\\
A_{2} &= \frac{\sqrt{\Delta}}{\sqrt{N}}\eta ,
\end{align*}
and $\eta$ is a mean $0$ multivariate Gaussian random variable with correlation matrix $\sigma^2$ defined as
\begin{equation*}
(\sigma^2)_{lm} = \sum_{j = 1}^{M} v_{jl}v_{jm}\int_{0}^{\infty} f_j(x_N((k-1)\Delta -s)) ~d\mu_j(s).
\end{equation*}
For $(k-1) \Delta < t < k \Delta $, define $x_{N} (t)$ by linearly interpolating between $x_{N} ((k-1) \Delta )$ and $x_{N} (k \Delta )$.
\end{definition}

Theorem~\ref{thm:tube} asserts that realizations of $b_{N}$ stay close to those of $x_{N}$ with high probability.

\begin{m_theorem}
\label{thm:tube}
Suppose that $\Delta (N) = N^{-1/4}$.  There exist constants $K_{1}$ and $K_{2}$ such that
\begin{equation}
\label{e:tube}
\mbb{P} \left( \vnorm{ b_{N} - x_{N} }_{\ell } > \frac{K_{1} T \zeta }{N^{1/8}} \right) \leqs K_{2} T e^{-N^{1/4}},
\end{equation}
where $\zeta $ is a system constant defined by
\begin{equation*}
\zeta = 2M \cdot \max_{1 \leqs j \leqs M} \vnorm{f_{j}}_{\infty } \cdot \max_{\substack{ 1 \leqs j \leqs M\\ 1 \leqs k \leqs D}} |v_{jk}|
\end{equation*}
and $\vnorm{ \cdot }_{\ell }$ is the `discretized' norm
\begin{equation*}
\vnorm{ b_{N} - x_{N} }_{\ell } = \max_{\substack{ k \in \mbb{Z} \\ 0 \leqs k \leqs T / \Delta }} \vabs{ b_{N} (k \Delta ) - x_{N} (k \Delta ) }.
\end{equation*}
\end{m_theorem}

\section{Discussion} \label{sec:discuss}

Stochastic differential equations (SDEs) are one of our main tools for modeling noisy processes in nature.  Interactions between the components of a system or network are frequently not instantaneous.  It is therefore natural to include such delay into corresponding stochastic models.  However, the relationship between delay SDEs (dSDEs) and the processes they model has not been fully established.

Delay stochastic differential equations have previously been formally derived from the delay chemical master equation (dCME).  Unlike the chemical master equation, however, the dCME is not closed; this complicates the derivation of dSDE approximations.  Closure in this context means the following:  Let $P(n,t)$ denote the probability that the stochastic system is in state $n$ at time $t$.  The dCME expresses the time derivative of $P(n,t)$ in terms of joint probabilities of the form $P(j,t;k,t - \tau )$ - the probability that the system is in state $j$ at time $t$ and was in state $k$ at time $t - \tau $, where $\tau $ is the delay.  The one-point probability distribution $P(\cdot , t)$ is therefore expressed in terms of two-point joint distributions, resulting in a system that is not closed.  Timescale separation assumptions have been used to close the dCME.  If the delay time is large compared to the other timescales in the system, one may assume that events that occur at time $t - \tau$ are decoupled from those that occur at time $t$ and close the dCME~\cite{Bratsun_Volfson_Tsimring_Hasty_2005, Tian_Burrage_Burrage_Carletti_2007} by assuming the joint probabilities may be written as products:
\begin{equation*}
P(j,t;k,t - \tau ) = P(j,t) P(k,t - \tau ).
\end{equation*}
Having closed the dCME, one may then derive dSDE approximations~\cite{Tian_Burrage_Burrage_Carletti_2007} as well as useful expressions for autocorrelations and power spectra~\cite{Bratsun_Volfson_Tsimring_Hasty_2005}.

Approximations of dSDE type have also been derived using system size expansions such as van Kampen expansions and Kramers-Moyal expansions for both fixed delay~\cite{Galla_2009} and distributed delay~\cite{Lafuerza_Toral_2011}.


Brett and Galla~\cite{Brett_Galla_2013} use the path integral formalism of Martin, Siggia, Rose, Janssen, and de Dominicis to derive the delay chemical Langevin equation (dCLE) without relying on the dCME.  Using this formalism, a moment generating functional may be expressed in terms of the system size parameter $N$ and the sampling rate $\Delta$.  In the continuous-time limit, $\Delta \to 0$, the dCLE may be inferred from the moment generating functional.   However, the Brett and Galla derivation has some limitations.  First, the $\Delta \to 0$ limit cannot be taken without simultaneously letting $N \to \infty$. Intuitively, this is because as $\Delta \to 0$, the Gaussian approximation to the Poisson distribution with mean $N \lambda \Delta$ breaks down unless the parameter $N \lambda$ simultaneously diverges to infinity.  Second, the derivation gives no quantitative information about the distance between the dCLE and the original delay birth-death (dBD) process. 

In this paper, we address these shortcomings. We prove rigorously that the dBD process can be approximated by a class of Gaussian processes that includes the dCLE.  In particular, we establish that for most biophysically relevant propensity functions, the dCLE process will approximate all moments of the dBD process.  The rigorous proof includes bounds on the quality of the approximation in terms of the time $T$ for which the approximation is desired to hold and the characteristic protein number $N$ (see Theorem \ref{thm:tube}).  The error bounds also indicate that the quality of the dCLE approximation worsens with increasing upper bounds on the reaction propensity functions and state-change vectors.  Physically, this means that high reaction rates and reactions that cause large changes in the protein populations are detrimental to the quality of the dCLE approximation.  

The dCLE is one of many Gaussian processes that approximate the dBD process.  Among all Gaussian approximations with noise components that scale as $1 / \sqrt{N}$, the dCLE is optimal because it is the only such approximation that exactly matches the first and second moments of the dBD process.  We formally justify this assertion in the supplement~\cite{supplementary_material} using characteristic functions.  As our simulations of the degrade and fire oscillator demonstrate, the dCLE can significantly outperform other Gaussian approximations at moderate system sizes.

Nevertheless, the quantitative tube estimates in Theorem~\ref{thm:tube} apply to {\itshape any} Gaussian approximation of the dBD process provided the noise scales as $1 / \sqrt{N}$.  This is significant because it is often advantageous to use linear noise approximations of the dCLE.  Delay appears in the drift component of a linear noise approximation but not in the diffusion component.  Linear noise approximations are therefore easier to analyze than their dCLE counterparts.  In particular, elements of the theory of large deviations for Markovian systems can be extended to SDEs with delay in the drift~\cite{schwartz2012noise}.

For metastable systems, our simulations indicate that the dCLE captures both temporal information (such as hitting times for the positive feedback model; see Fig.~\ref{fig:pf_hitting}) and spatial information (such as densities for trajectories corresponding to failed and successful transitions; see Fig.~\ref{fig:toggle}).  This suggests that dCLE approximations may be used to study rare events for biochemical systems that exhibit metastability.

We have shown that the dCLE provides an accurate approximation of a number of stochastic processes.  Although we chose gene regulatory networks in our examples, the theory is applicable to general birth-death processes with delayed events.  SDEs, and the chemical Langevin equation in particular, are fundamental in modeling and understanding the behavior of natural and engineered systems.  We therefore expect that the dCLE will be widely applicable when delays impact system dynamics.

\appendix

\section{Setting and main results}
\label{s:main_results}

Consider a system of $D$ biochemical species and $M$ possible reactions.  We are interested in describing the dynamics as a function of a latent system parameter $N$, the system size.  Let $B_{N} (t) \in \mbb{Z}^{D}$ denote the state of the system at time $t$.

Each reaction $R_{j}$ is described by the following:
\begin{enumerate}[labelindent=\parindent, label=(\alph*), ref=\alph*]

\item
A propensity function $f_{j} : \mbb{R}^{D} \to \mbb{R}^{+}$.  The firing rate of reaction $R_{j}$ is given by $N f_{j} (B_{N} (t) / N)$.

\item
A state-change vector $\mbi{v}_{j} \in \mbb{Z}^{D}$.  The vector $\mbi{v}_{j}$ describes the change in the number of molecules of each species that results from the completion of a reaction of type $j$.

\item
A probability measure $\mu_{j}$ supported on $[0, \tau_{0}]$.  The measure $\mu_{j}$ models the delay that may occur between the initiation and completion of a reaction of type $j$.  If reaction $R_{j}$ is instantaneous, then $\mu_{j} = \delta_{0}$.  If the delay is a fixed value $\tau > 0$, then $\mu_{j} = \delta_{\tau }$.  If $\mu_{j}$ has a density, we denote it by $\kappa_{j}$.
\end{enumerate}

Given a system trajectory up to time $t$, $\vset{ B_{N} (s) : s \leqs t }$, the dSSA can be described as follows:

\begin{enumerate}[leftmargin=*, label=(I\arabic*), ref=I\arabic*]

\item
\label{i:time_reaction}
Sample a time $\xi $ to the next reaction from the exponential distribution with rate
\begin{equation*}
N \sum_{j=1}^{M} f_{j} \left( \frac{ B_{N} (t) }{ N } \right).
\end{equation*}

\item
Select a reaction $R_{k}$  with probability
\begin{equation*}
\frac{ f_{k} (B_{N} (t) / N) }{ \sum_{j=1}^{M} f_{j} (B_{N} (t) / N) }.
\end{equation*}

\item
Sample a delay time $\tau $ from $\mu_{k}$.

\item
If no delayed reactions are set to finish in the time interval $[t, t + \xi ]$, then proceed as follows.  If $\tau = 0$, then move to time $t + \xi $ and set $B(t + \xi ) = B(t) + \mbi{v}_{k}$.  If $\tau > 0$, then move to time $t + \xi $, set $B(t + \xi ) = B(t)$, and put the state change vector $\mbi{v}_{k}$ into a queue along with the designated time of reaction completion, $t + \xi + \tau $.

\item
If a reaction from the past is set to complete in $[t, t + \xi ]$, then move to the time $\hat{t}$ of the completion of the first such reaction, update $B(\hat{t})$ accordingly, and then proceed to~\pref{i:time_reaction}.

\end{enumerate}

Schlicht and Winkler~\cite{Schlicht_Winkler_2008} prove the existence of the stochastic process $(B_{N} (t))$ from which the dSSA samples.  Let $\mbb{B}_{N}$ denote the probability measure on realizations $\omega_{B_{N}/N} (t)$ associated with the scaled process $(B_{N} (t) / N)$.  These realizations lie in the space $\mcal{D}$ of right-continuous functions from $[0,T]$ into $\mbb{R}^{D}$ that possess limits from the left.  Our main results quantify the behavior of $\mbb{B}_{N}$ as $N \to \infty $.  We make the following regularity assumptions on the propensities $f_{j}$.

\begin{enumerate}[leftmargin=*, label=(P\arabic*), ref=P\arabic*]

\item
\label{i:propen_reg}
The functions $f_{j}$ have continuous derivatives of order $2$.

\item
\label{i:propen_supp1}
There exists a compact set $K$ such that $\supp (f_{j}) \subset K$ for all $1 \leqs j \leqs M$.

\item
\label{i:propen_supp2}
For all $1 \leqs j \leqs M$, $f_{j} (\mbi{x}) > 0$ only if all coordinates of the vector $\mbi{v}_{j} + \mbi{x}$ are nonnegative; $f_{j} (\mbi{x}) = 0$ otherwise.

\end{enumerate}

The correct dCLE approximation of $(B_{N} (t) / N)$ is given for $1 \leqs k \leqs D$ by
\begin{equation}
\label{e:dcle_general}
\mrm{d} x_{k} = \left( \sum_{j=1}^{M} \int_{0}^{\tau_{0} } v_{jk} f_{j} (\mbi{x} (t - s)) \, \mrm{d} \mu_{j} (s) \right) \mrm{d} t + \frac{1}{\sqrt{N}} (\Sigma \, \mathrm{d} \msf{W})_{k},
\end{equation}
where $\msf{W}$ is a $D$-dimensional vector of independent standard Brownian motions and $\Sigma^{2}$ is given by
\begin{equation*}
\Sigma^{2}_{lm} = \sum_{j=1}^{M} v_{jl} v_{jm} \int_{0}^{\tau_{0} } f_{j} (\mbi{x} (t-s)) \, \mathrm{d} \mu_{j} (s).
\end{equation*}

Let $L_{N}$ denote the stochastic process described by Eq.~\eqref{e:dcle_general} and let $\mbb{L}_{N}$ denote the probability measure on realizations $\omega_{L_{N}} (t)$ associated with this process.  Theorem~\ref{thm:dssa_dcle} controls the distance between $\mbb{B}_{N}$ and $\mbb{L}_{N}$.  

\begin{m_theorem}
\label{thm:dssa_dcle}
Assume that the propensities $f_{j}$ satisfy~\pref{i:propen_reg}--\pref{i:propen_supp2}.  Fix $T > 0$.  For every continuous observable $\Psi : \mcal{D} \to \mbb{R}$, we have
\begin{equation*}
\int_{\mcal{D}} \Psi (\omega_{B_{N}/N}) \, \mathrm{d} \mbb{B}_{N} - \int_{\mcal{D}} \Psi (\omega_{L_{N}}) \, \mathrm{d} \mbb{L}_{N} \to 0
\end{equation*}
as $N \to \infty $.
\end{m_theorem}

Theorem~\ref{thm:dssa_therm} establishes weak convergence of the scaled process $(B_{N} (t) / N)$ to the thermodynamic limit governed by the delay reaction rate equations
\begin{equation}
\label{e:drre_general}
\mrm{d} x_{k} = \left( \sum_{j=1}^{M} \int_{0}^{\tau_{0} } v_{jk} f_{j} (\mbi{x} (t - s)) \, \mrm{d} \mu_{j} (s) \right) \mrm{d} t.
\end{equation}

\begin{m_theorem}
\label{thm:dssa_therm}
Assume that the propensities $f_{j}$ satisfy~\pref{i:propen_reg}--\pref{i:propen_supp2}.  Fix $T > 0$.  For every continuous observable $\Psi : \mcal{D} \to \mbb{R}$, we have   
\begin{equation*}
\int_{\mcal{D}} \Psi (\omega_{B_{N}/N}) \, \mathrm{d} \mbb{B}_{N} \to \Psi (\mbi{x})
\end{equation*}
as $N \to \infty $, where $\mbi{x}$ denotes the solution of Eq.~\eqref{e:drre_general}.
\end{m_theorem}

Crucial to the proofs of Theorems~\ref{thm:dssa_dcle} and~\ref{thm:dssa_therm} are the discretization of the time interval $[0,T]$ and the development of quantitative controls on processes that approximate $b_{N} \defas B_{N} / N$.  We partition $[0,T]$ into subintervals of length $\Delta = \Delta (N)$.  It is crucial that $\Delta (N)$ scales with $N$ in the right way.  The quantitative pathwise controls in Theorem~\ref{thm:tube} hold if $\Delta (N) = N^{-1/4}$.

We define an Euler-Maruyama discretization $x_{N}$ of~\eqref{e:dcle_general}.

\begin{definition}[The process $x_N(t)$] 

For $t \leqs 0$, define $x_N(t) = b_N(t)$. For integers $k \geqs 1$ and $t = k \Delta $, define $x_N(t)$ recursively by
\begin{equation*}
x_N (k\Delta) = x_N((k-1)\Delta) + \left[ \sum_{j = 1}^{M} v_j \int_{0}^{\infty} f_j(x_N((k-1)\Delta -s))~d\mu_j(s) \right] \Delta + \frac{\sqrt{\Delta}}{\sqrt{N}}\eta ,
\end{equation*}
where $\eta$ is a mean $0$ multivariate Gaussian random variable with correlation matrix $\sigma^2$ defined as
\begin{equation*}
(\sigma^2)_{lm} = \sum_{j = 1}^{M} v_{jl}v_{jm}\int_{0}^{\infty} f_j(x_N((k-1)\Delta -s)) ~d\mu_j(s).
\end{equation*}
For $(k-1) \Delta < t < k \Delta $, define $x_{N} (t)$ by linearly interpolating between $x_{N} ((k-1) \Delta )$ and $x_{N} (k \Delta )$.
\end{definition}

Theorem~\ref{thm:tube} asserts that realizations of $b_{N}$ stay close to those of $x_{N}$ with high probability.

\begin{m_theorem}
\label{thm:tube}
Suppose that $\Delta (N) = N^{-1/4}$.  There exist constants $K_{1}$ and $K_{2}$ such that
\begin{equation}
\label{e:tube}
\mbb{P} \left( \vnorm{ b_{N} - x_{N} }_{\ell } > \frac{K_{1} T \zeta }{N^{1/8}} \right) \leqs K_{2} T e^{-N^{1/4}},
\end{equation}
where $\zeta $ is a system constant defined by
\begin{equation*}
\zeta = 2M \cdot \max_{1 \leqs j \leqs M} \vnorm{f_{j}}_{\infty } \cdot \max_{\substack{ 1 \leqs j \leqs M\\ 1 \leqs k \leqs D}} |v_{jk}|
\end{equation*}
and $\vnorm{ \cdot }_{\ell }$ is the `discretized' norm
\begin{equation*}
\vnorm{ b_{N} - x_{N} }_{\ell } = \max_{\substack{ k \in \mbb{Z} \\ 0 \leqs k \leqs T / \Delta }} \vabs{ b_{N} (k \Delta ) - x_{N} (k \Delta ) }.
\end{equation*}
\end{m_theorem}

\section{Proofs}
\label{s:proofs}

\subsection{Overview}
\label{ss:overiew}

Theorem~\ref{thm:dssa_therm} establishes weak convergence of the scaled process $(b_{N} (t))$ to the thermodynamic limit~\eqref{e:drre_general}.  We prove Theorem~\ref{thm:dssa_therm} in two steps.  First, we show that the family of measures $(\mbb{B}_{N})$ is a tight family on $\mcal{D}$.  By the Prohorov theorem (see \textit{e.g.}~\cite{Billingsley_1999}), the family $(\mbb{B}_{N})$ is then relatively compact in the space of probability measures.  Second, we consider finite sequences of times $t_{1} < t_{2} < \cdots < t_{k}$ in $[0,T]$ and study the finite-dimensional distributions associated with $(b_{N} (t_{1}), \ldots , b_{N} (t_{k}))$ in order to show that the family $(\mbb{B}_{N})$ has a weak limit and to characterize this limit.

We develop pathwise tube estimates to complete the second step of the proof of Theorem~\ref{thm:dssa_therm}.  We use these tube estimates to prove Theorem~\ref{thm:dssa_dcle} as well.

{\bfseries Scaling of the time discretization} $\bsym{\Delta }${\bfseries .}
We assume throughout that $\Delta = \Delta (N)$ satisfies $C_{1} N^{-\alpha } \leqs \Delta (N) \leqs C_{2} N^{-\alpha }$, where $C_{1}$ and $C_{2}$ are constants and $1/4 \leqs \alpha < 1/3$.  Intuitively, $\Delta $ must be sufficiently small so that propensity functions do not change significantly over any time interval of length $\Delta $ and sufficiently large so that many reactions fire over any such interval.

\subsection{Tightness}
\label{ss:tightness}

The space $\mcal{D}$ is a metric space with Skorohod metric $d$ defined as follows.  Let $\Gamma $ denote the set of nondecreasing functions $\gamma : [0,T] \to [0,T]$ with $\gamma (0) = 0$ and $\gamma (T) = T$.  For $\gamma \in \Gamma $, define
\begin{equation*}
\vnorm{\gamma } = \sup_{s \neq t} \vabs{ \ln \left( \frac{\gamma (t) - \gamma (s)}{t-s} \right) }.
\end{equation*}
For functions $z_{1}$ and $z_{2}$ in $\mcal{D}$, define
\begin{align*}
d(z_{1}, z_{2}) = \inf \Big\{ \vep > 0 : &\text{ there exists } \gamma \in \Gamma \text{ with } \vnorm{ \gamma } \leqs \vep 
\\
&\qquad \qquad \qquad \text{and } \sup_{t \in [0,T]} \vabs{z_{1} (t) - z_{2} (\gamma (t)) } \leqs \vep \Big\}.
\end{align*}

Let $(\mu_{n})$ be a sequence of probability measures on $\mcal{D}$.  We recall a characterization of tightness for $(\mu_{n})$.  For $z \in \mcal{D}$ and $0 < \delta < T$, define
\begin{equation*}
w_{z} (\delta ) = \inf_{\vset{t_{i}}} \max_{i} \sup_{s,t \in [t_{i-1}, t_{i})} \vabs{ z(s) - z(t) },
\end{equation*}
where the infimum is taken over partitions of $[0,T]$ such that $t_{i} - t_{i-1} > \delta $ for all $i$.

\begin{proposition}[\cite{Billingsley_1999}]
\label{p:tightD}
The sequence $(\mu_{n})$ is tight if and only if the following hold.
\begin{enumerate}[leftmargin=*, label=(T\arabic*), ref=T\arabic*]

\item
\label{i:tight1}
For every $\eta > 0$, there exists $a$ such that
\begin{equation*}
\mu_{n} ( \{ z : \sup_{t} |z(t)| > a \} ) \leqs \eta
\end{equation*}
for all $n \in \mbb{N}$.

\item
\label{i:tight2}
For all $\vep > 0$ and $\eta > 0$, there exist $0 < \delta < T$ and $n_{0} \in \mbb{N}$ such that
\begin{equation*}
\mu_{n} ( \{ z : w_{z} (\delta ) \geqs \vep \} ) \leqs \eta
\end{equation*}
for all $n \geqs n_{0}$.

\end{enumerate}
\end{proposition}

We now prove that $(\mbb{B}_{N})$ is tight.

\begin{lemma} 
\label{lemma:A1}
Suppose $X$ is a Poisson random variable with parameter $\lambda >
0$.  For every $a > \lambda$, we have
\begin{equation}
\label{e:poisson}
\mbb{P} (X > a) \leq \frac{\exp(a - \lambda)}{\left( \frac{a}{\lambda}\right)^a}.
\end{equation}
\end{lemma}

\begin{proof}
It follows from the Markov inequality that 
\[
\mbb{P} (X > a) \leq \inf_{t>0} \mbb{P} (e^{tX} > e^{ta}) \leq \inf_{t>0} M_X(t)/e^{at} = \inf_{t>0} e^{\lambda (e^t -1) - at}
\]
and so it suffices to minimize the function $\phi(t) \defas \lambda(e^{t}
-1) - at.$ It is easy to see that the minimum is reached at $t =
\ln(a/\lambda)$, provided $a > \lambda$. On substituting back into
$\exp(\lambda(e^t -1) - at)$ we get the result.
\end{proof}

The bound from Lemma~\ref{lemma:A1} is one that we will encounter often. The next lemma simplifies the bound.

\begin{lemma} 
\label{lemma:A2}
(Simplification of~\eqref{e:poisson} in Lemma~\ref{lemma:A1})
\begin{enumerate}[labelindent=\parindent, label=(\alph*), ref=\alph*]
\item Suppose $0 < y < x$. There exists a constant $C_3(x,y)$ such that 
\[
\frac{\exp\left( n [x -y] \right)}{\left(\frac{x}{y}\right)^{nx}}  \leq  e^{-C_3 ny} \to 0 \quad \text{ as } n\to \infty
\]

\item Let $y > 0$, $\epsilon > 0$ and $n\in\mathbb{N}$. Let $x> y +\epsilon.$ There exists a constant $C_4 (y, \epsilon)$ such that 
\[
\frac{\exp\left( n [x -y] \right)}{\left(\frac{x}{y}\right)^{nx}}  \leq  e^{-C_4 nx} \to 0 \quad \text{ as } x\to \infty
\]
\end{enumerate}
\end{lemma}

\begin{proof}
First inequality: Since $0< y < x$, the expression above is an indeterminate form of
the type $\infty/\infty$. On computing logarithms, we need to show that
the expression below diverges to $-\infty$:
\[
\phi(n):= n x - ny - n x \ln(x/y) = ny \left[ \frac{x}{y} -1 - \log \left( \frac{x}{y}\right) \right]
\]
This happens if and only if $z -1 < z \log (z)$ for $z = x/y$, which, in turn, is always true when $z > 1$.
The proof of the second inequality is similar.
\end{proof}

Central to the arguments that follow is a bound on the probability that $\vabs{ b_{N} (t_{0} + \Delta ) - b_{N} (t_{0}) } \geqs \Delta \zeta $, where $\zeta $ is appropriately chosen and the bound is uniform in $t_{0}$.  Define
\begin{equation*}
\zeta = 2VM \cdot \max_{i} \vabs{ f_{i} }_{\infty },
\end{equation*}
where $V = \max_{j,k} \vabs{ v_{jk} }$.

\begin{lemma}
\label{l:firing_bound}
Let $E = E(i_{0},t_{0},\Delta )$ be the event that reaction $i_{0}$ fires at least $N \Delta \zeta / (VM)$ times on $[t_{0}, t_{0} + \Delta ]$.  Using Lemma~\ref{lemma:A1} and Lemma~\ref{lemma:A2}, we have
\begin{equation}\label{supp:eq:stretch_exp}
\mbb{B}_{N} (E) \leq \frac{\exp \left( N \Delta [2 \cdot \max\set{|f_i|_\infty} - |f_{i_{0}}|_\infty]\right)}{\left( \frac{2 \cdot \max\set{|f_i|_\infty}}{|f_{i_{0}}|_\infty}\right)^{2N \Delta \cdot \max\set{|f_i|_\infty} }} \leq e^{-2 C_4 N \Delta \max_{i} |f_i|_\infty} \defasr e^{-C_5 N^{1-\alpha}}.
\end{equation}
\end{lemma}

\begin{corollary}
\label{c:c_jump_bound}
For all $t_{0} \in [0, T - \Delta ]$, we have
\begin{equation}
\label{e:c_jump_bound}
\mbb{P} (|b_{N} (t_{0} + \Delta ) - b_{N} (t_{0})| \geqs \Delta \zeta ) \leqs M e^{-C_5 N^{1-\alpha}}.
\end{equation}
\end{corollary}

\begin{proposition}
\label{p:tightness}
The sequence $(\mbb{B}_{N})$ is tight.
\end{proposition}

\begin{proof}
We use Proposition~\ref{p:tightD}.  Condition~\pref{i:tight1} follows from the assumptions on the propensity functions.  Condition~\pref{i:tight2} follows from~\eqref{e:c_jump_bound}.
\end{proof}

\subsection{Tube estimates via characteristic functions}
\label{ss:tubes}

We begin with a technical estimate.

\begin{lemma}\label{dCLE:supp:exponentials_are_close} For every constant $C>0$,
for $N\in\mathbb{N}$ large enough,
\begin{equation*}
\begin{aligned}
\left|
\exp\left(
N\Delta \left[
\sum_{j=1}^M \left( e^{\mathrm{i}\mathbf{r}v_j/N}-1\right) \int_{0}^{\infty}
[f_j(b_N(t -s)) - C\Delta\zeta] ~d\mu_j(s)
\right]
\right) \right.\\
- \left.
\exp\left(
N\Delta \left[
\sum_{j=1}^M \left( e^{\mathrm{i}\mathbf{r}v_j/N}-1\right) \int_{0}^{\infty}
[f_j(b_N(t -s)) + C\Delta\zeta] ~d\mu_j(s)
\right]
\right)
\right|\\
\hfill \leq 2C \Delta^2 \zeta |\mathbf{r}| \left(\sum_{j = 1}^M |v_j|\right) + 
\frac{C}{2} |\mathbf{r}|^2 \Delta^3 \zeta \left( \sup_{j} M |f_j|_\infty |v_j|\right)^2.
\end{aligned}
\end{equation*}
\end{lemma}

\begin{proof}
\begin{equation*}
\begin{aligned}
\left|
\exp\left(
N\Delta \left[
\sum_{j=1}^M \left( e^{\mathrm{i}\mathbf{r}v_j/N}-1\right) \int_{0}^{\infty}
[f_j(b_N(t -s)) - C\Delta\zeta] ~d\mu_j(s)
\right]
\right) \right.\\
- \left.
\exp\left(
N\Delta \left[
\sum_{j=1}^M \left( e^{\mathrm{i}\mathbf{r}v_j/N}-1\right) \int_{0}^{\infty}
[f_j(b_N(t -s)) + C\Delta\zeta] ~d\mu_j(s)
\right]
\right)
\right| \\
\leq \left|
\exp\left( 
N\Delta \sum_{j = 1}^M \left\{\int_{0}^{\infty} \left[f_j(b_N(t -s)) - C\Delta \zeta\right]~ d\mu_j(s)\right\} \left( \frac{\mathrm{i} \mathbf{r} v_j}{N} - \frac{(\mathbf{r}v_j)^2}{N^2} + \mathcal{O}(N^{-3})\right)
\right)\right.\\
 \left. -
 \exp\left( 
N\Delta \sum_{j = 1}^M \left\{\int_{0}^{\infty} \left[f_j(b_N(t -s)) + C\Delta \zeta\right]~ d\mu_j(s)\right\} \left( \frac{\mathrm{i} \mathbf{r} v_j}{N} - \frac{(\mathbf{r}v_j)^2}{N^2} + \mathcal{O}(N^{-3})\right)
\right)
\right|\\
\leq \left|
N\Delta \sum_{j = 1}^M \left( \frac{\mathrm{i}\mathbf{r}v_j}{N} - \frac{(\mathbf{r}v_j)^2}{N^2} + \mathcal{O}(N^{-3})\right) \left( \int_{0}^{\infty} 2 C\Delta \zeta ~ d\mu_j(s)\right)
\right|\\
+ \frac{1}{2} \left|
N^2\Delta^2 \left\{ 
\sum_{j = 1}^M \left( \frac{\mathrm{i}\mathbf{r}v_j}{N} - \frac{(\mathbf{r}v_j)^2}{N^2} + \mathcal{O}(N^{-3})\right)^2\right.\right.\\
\left.\left. \left\{
\int_{0}^\infty \left[(f_j(b_N(t - s)) - C \Delta \zeta)^2 -(f_j(b_N(t - s)) + C \Delta \zeta)^2 \right] d\mu_j(s)
\right\}
\right\}
\right|\\
\leq  2C \Delta^2 \zeta |\mathbf{r}| \left( \sum_{j = 1}^{M} |v_j|\right) + \frac{C}{2}|\mathbf{r}|^2 \Delta^3 \zeta \left( \sum_{j = 1}^M \int_{0}^\infty |v_j f_j(t - s)|~ d\mu_j(s)  \right)^2
\end{aligned}\phantom{\hfill}
\end{equation*}
\end{proof}

\begin{proposition}\label{dCLE:prop:3}
Define $\varphi = B_{N} (t_{0} + \Delta ) - B_{N} (t_{0})$.  For $\mathbf{r}\in \mathbb{R}^D$ and $\epsilon > 0$, we have
\begin{equation}
\begin{aligned}
\left|E\left[
e^{i \mathbf{r} \varphi /N} 
\right] 
- 
E\left[
\exp\left( 
N \Delta \sum_{j = 1}^M\left( e^{\mathrm{i}\mathbf{r} v_j /N} -1\right)  \int_{0}^{\infty}   f_j (b_N(t_0 -s)) ~d\mu(s) 
\right)
\right]
\right|\\
\leq 5 |\mathbf{r}| \Delta^2 \zeta |f'|_\infty + o(\Delta^{2+\epsilon})
\end{aligned}
\end{equation} where 
\begin{equation}
|f'|_\infty \defas \left. \sup_{j} \sup_{t\geq 0}\frac{d}{d x_t} \int_{0}^\infty  f_j(x_{t-s})~d\mu_j(s)\right|_{x_t = b_{N} (t)} \defasr \nabla_{N}.
\end{equation} 
Since  $f_i \in \mathcal{C}^2$, we may also take $|f'|_\infty  = \sup_{j} |f_j'|_\infty.$
\end{proposition}

\begin{proof}

Corollary~\ref{c:c_jump_bound} gives
$\mathbb{B}_N \set{|\varphi |/N > \Delta\zeta} \leqs M e^{-C_5 N^{1-\alpha}}$. 
Define a random variable 
$\varphi_{(k)}$ as follows:
\begin{equation*}
\varphi_{(k)} = \begin{cases}
\varphi & \text{ if $|\varphi_{i}| < k$ for $1\leq i \leq D$}\\
0 & \text{ if $|\varphi_{i}|  \geq k$ for some $1 \leq i \leq D$}
\end{cases}.
\end{equation*} 
Abusing notation, we write $\varphi_{k}$ for $\varphi_{(k)}$ throughout the proof.  One can see that $\varphi_k \to \varphi$ pointwise, and by the  
dominated convergence theorem we
can directly obtain that $E[\varphi_k] \to E[\varphi]$. However, in doing so we
don't obtain a rate. For this reason, we perform a more careful computation.

Define sets $\set{\varphi \leq l}$ as the set of all outcomes for which $\varphi_{(i)} \leq l$ for all $1 \leq i \leq D$. The complement of this set is denoted as $\set{\varphi > l}$. If For any $l \in \mathbb{N}$, 
$\set{N l\zeta \Delta \leq \varphi  \leq N(l+1)\zeta \Delta} \subset \set
{\varphi \geq N l\zeta \Delta} $, and so from Equation \eqref{supp:eq:stretch_exp} 
it follows that $\mathbb{B}_N \set{N l \zeta \Delta \leq \varphi \leq N (l+1) 
\zeta \Delta} \leq e^{-C_4 N\Delta l\zeta}$ (by Lemma \ref{lemma:A2}). It then follows that
\begin{equation*}
\begin{aligned}
\int_{\varphi \leq N\zeta\Delta} \varphi + \sum_{l \geq 1} lN\Delta\zeta e^{-C_4lN\Delta\zeta} \leq \int \varphi \leq \int_{\varphi \leq N\zeta\Delta} \varphi + \sum_{l \geq 1} (l+1)N\Delta\zeta e^{-C_4lN\Delta\zeta}
\end{aligned}
\end{equation*} and so $|E[\varphi] - E[\varphi_{N\Delta\zeta}]| \leq \zeta\sum_{l \geq 1} (l+1)N\Delta e^{-C_4l N \Delta\zeta}.$ The bound on the right is smaller than $2 \zeta e^{-C_4 \zeta N^{1 - \alpha}}.$ An analogous computation can be used to show that $|E[e^{\mathrm{i}\mathbf{r} \varphi/N}] - E[e^{\mathrm{i}\mathbf{r} \varphi_{N\Delta\zeta}/N}]| \leq 4 \zeta e^{-C_4 \zeta N^{1 - \alpha}}.$ Therefore, it is enough to approximate $E[e^{\mathrm{i}\mathrm{r} \varphi_{k}}]$ for $k = N\zeta\Delta$.

Let $\tilde{\varphi}$ be the random variable that defines the change to a process over the interval $[t_0, t_0 + \Delta]$ that has constant propensity
functions $\int_{0}^{\infty}f_j(b_N(t_0 -s))~ d\mu_j(s)$ for the reaction that contributes a change $v_j$ to the system. Call $\tilde{P_N}$ the stationary measure for such a process. It follows from definition that
\begin{equation}
\begin{split}
\left| 
E[e^{\mathrm{i}\mathbf{r} \varphi_k/N}] - E[e^{\mathrm{i}\mathbf{r} \tilde{\varphi}/N}]
\right| \leq \bigg|
\int_{\varphi \leq k} e^{\mathrm{i}\mathbf{r} \varphi_k/N}~ d\mathbb{B}_N - 
&\int_{\varphi \leq k} e^{\mathrm{i}\mathbf{r}\tilde{\varphi}/N}~d\tilde{P_N} \bigg| \\
&{}+ \sum_{l \geq k+1} \tilde{P_N}(l) + \sum_{l \geq k+1} \mathbb{B}_N(l).
\end{split}
\end{equation} $\tilde{P_N}(l)$ is abbreviation for $\tilde{P_N} \set{\tilde l-1 < {\varphi} \leq l}$ (and analogously for $\mathbb{B}_N(l)$).

We now compute bounds for the first term on the right. Recall that $k = N\Delta\zeta$. Since the largest size of any co-ordinate of  $\tilde{\varphi}$ is 
smaller than $N\Delta\zeta$, the arguments to the propensity functions change by at most $\Delta\zeta$. From the smoothness of the propensity functions it follows
that for any $t_0 \leq t  \leq t_0+\Delta$, 
\begin{equation*}
\begin{aligned}
\Delta \int_{0}^{\infty} f_j(b_N(t_0 -s)) ~ d\mu_j(s) - \Delta^2 \zeta|f'|_\infty
\leq \Delta \int_{0}^{\infty} f_j(b_N(t-s))~ d\mu_j(s) \\
\leq 
\Delta \int_{0}^{\infty} f_j(b_N(t_0 -s))~d\mu_j(s) + \Delta^2 \zeta |f'|_\infty.
\end{aligned}
\end{equation*} The estimate in the middle is the infinitesimal rate for the process
$\varphi_k$, and can be approximated only in terms of information available up to
time $t_0$. Therefore, 
\begin{equation}
\begin{aligned}
\left|
\int_{\varphi \leq k} e^{\mathrm{i}\mathbf{r} \varphi_k/N}~d\mathbb{B}_N - \int_{\varphi \leq k} e^{\mathrm{i}\mathrm{r} \tilde{\varphi}/N}~d\tilde{P_N}
\right| &\leq \Big|
E\left[
e^{N\Delta \sum_{j = 1}^M(e^{\mathrm{i}\mathbf{r}v_j/N} -1)(\int_{0}^{\infty} f_j(b_N(t_0 -s)) ~ d\mu_j(s) - \Delta \zeta |f'|_\infty)}
\right] \\
&{}- E\left[
e^{N\Delta \sum_{j = 1}^M(e^{\mathrm{i}\mathbf{r}v_j/N} -1)(\int_{0}^{\infty} f_j(b_N(t_0 -s)) ~ d\mu_j(s) + \Delta \zeta |f'|_\infty)}
\right]
\Big| \\
&\leq 5 |f'|_\infty \Delta^2 \zeta |\mathbf{r}| \left(\sum_{j = 1}^M |v_j|\right)
+ \mathcal{O}(\Delta^{2+\epsilon})
\end{aligned}
\end{equation} 
by Lemma~\ref{dCLE:supp:exponentials_are_close}.
The other terms have tails of the order of $e^{-N^{1 - \alpha}}$
and so, for large enough $N$, the dominant term is the one involving $\Delta^2$. Finally, we observe that 
\begin{equation}
E\left[e^{\mathrm{i}\mathbf{r}\tilde{\varphi}/N}\right] = E\left[
\exp\left( 
N \Delta \sum_{j = 1}^M\left( e^{\mathrm{i}\mathbf{r} v_j /N} -1\right)  \int_{0}^{\infty}   f_j (b_N(t_0 -s)) ~d\mu(s) 
\right)
\right].
\end{equation}
\end{proof}

Consider the approximation to the characteristic function for the jump $\varphi$ on the interval $[t_0, t_0+\Delta]$. On expanding the term $e^{\mathrm{i}\mathbf{r}v_j/N -1}$ term using a Taylor series expansion, we can approximate the expression 
\begin{equation} \label{dCLE:supp:e:FT_SSA}
E\left[
\exp\left(
N\Delta \sum_{j = 1}^{M} (e^{\mathrm{i}\mathbf{r}v_j/N} -1 ) \int_{0}^{\infty} f_j(b_N(t_0 -s))~d\mu_j(s)
\right)
\right]
\end{equation} by
\begin{equation}\label{dCLE:supp:e:FT_CLE}
E \left[
\exp\left(
\Delta \sum_{j = 1}^{M} \mathrm{i}\mathbf{r} v_j \int_{j = 0}^{\infty} f_j(b_N(t_0 -s))~d\mu_j(s) - \frac{\Delta}{N} \sum_{j = 1}^{M} (\mathbf{r} v_j)^2 \int_{0}^{\infty} f_j(b_N(t_0 -s))~ d\mu_j(s)
\right)
\right]
\end{equation} with an error that is $\mathcal{O}(N^{-2})$. However, this is the characteristic function for the Gaussian random variable
\begin{equation}
\begin{aligned}
\delta \mathbf{x}_t = \left[
\sum_{j = 1}^{M} v_j \int_{0}^{\infty} f_j(b_N(t - s)) ~ d\mu_j(s)
\right] \Delta + \frac{\sqrt{\Delta}}{\sqrt{N}}  \eta
\end{aligned}
\end{equation}\label{dCLE:supp:e:sigma}
where $\eta$ is a mean $0$ multivariate Gaussian random variable with correlation matrix $\sigma^2$, with
\begin{equation}
(\sigma^2)_{l,m} = \sum_{j = 1}^{M} v_{jl}v_{jm}\int_{0}^{\infty} f_j(b_N(t-s))~d\mu_j(s).
\end{equation}

This suggests the construction for the following approximating Gaussian process, defined as a discrete stochastic differential equation with Gaussian jumps. 
\begin{definition}[The process $x_N(t)$] \label{dCLE:supp:defn:dcle}
For $t \leq 0$, define $x_N(t) = b_N(t)$. For $t = k\Delta, k\geq 1$, define $x_N(t)$ recursively as
\begin{equation*}
x_N(k\Delta) = x_N((k-1)\Delta) + \left[
\sum_{j = 1}^{M} v_j \int_{0}^{\infty} f_j(x_N((k-1)\Delta -s))~d\mu_j(s)
\right] \Delta + \frac{\sqrt{\Delta}}{\sqrt{N}}\eta
\end{equation*} where $\eta$ is a mean 0 multivariate Gaussian random variable with correlation matrix
$\sigma^2$ defined as
\begin{equation*}
(\sigma^2)_{lm} = \sum_{j = 1}^{M} v_{jl}v_{jm}\int_{0}^{\infty} f_j(x_N((k-1)\Delta -s)) ~d\mu_j(s)
\end{equation*}
\end{definition}

\begin{definition}[The process $y_N(t)$]
For $t  \leq 0$ define the process $y_N(t) = b_N(t)$. For $t = k\Delta, k\geq 1$, define $y_N(t)$
recursively as 
\begin{equation*}
y_N(k\Delta) = y_N((k-1)\Delta) + \left[
\sum_{j = 1}^{M} v_j \int_{0}^{\infty} f_j(y_N((k-1)\Delta -s))~d\mu_j(s)
\right] \Delta.
\end{equation*}
\end{definition}

The next proposition estimates the probability of finding the processes $b_{N}$ and $x_N$ outside a tube around $y_N$ of radius greater than $4 T\zeta N^{- \frac{1}{2} - \frac{3\alpha}{2}}$. As the proposition shows, the Gaussian process has a smaller tail (of the order $N^{\alpha} e^{-N^{2\alpha}}$) than the birth-death process (which has a tail of order $N^{\alpha (1 - N^{\alpha})}$).

\begin{proposition}[Pathwise control on $b_N, x_N$ and $y_N$] \label{dCLE:supp:prop:pathwise}
For $1 \leq k \leq (T/\Delta)$, we have 
\[
\mathbb{B}_N \left(
\|b_N(k\Delta) - y_N(k\Delta)\|  > \frac{4k \zeta}{\sqrt{N\Delta}}
\right) \leq k \Delta^{\Delta^{-1}}
\] and 
\[
\mathbb{X}_N \left(
\|x_N(k\Delta) - y_N(k\Delta)\|  > \frac{4k \zeta}{\sqrt{N\Delta}}
\right) \leq k e^{-\Delta^{-2}/2}.
\]
\end{proposition}

\begin{proof}
Let $\varphi$ denote the increment to the process $b_N$ in the time interval $[0, \Delta]$. From Lemma~\ref{lemma:A1} it follows that for $N$ large enough
\[
\mathbb{B}_N \left(
\left\|
\varphi - \left[
\sum_{j = 1}^{M} v_j \int_{0}^{\infty} f_j(b_N(-s))~d\mu_j(s)
\right]\Delta
\right\| > \frac{\zeta}{\sqrt{N\Delta}}
\right) \leq \Delta^{\Delta^{-1}}.
\] Let $\delta y_N(0)$ denote the increment to the process $y_N$ on $[0,\Delta].$ Since $y_N = b_N$ for
$t \leq 0$, it follows that 
\[
\mathbb{B}_N \left(
\left\|
b_N(\Delta) - y_N(\Delta)
\right\| > \frac{\zeta}{\sqrt{N\Delta}}
\right) = \mathbb{B}_N \left(
\left\|
\varphi - \delta y_N(0)
\right\| > \frac{\zeta}{\sqrt{N\Delta}}
\right) \leq \Delta^{\Delta^{-1}}.
\] We now bound the error on the interval $[\Delta, 2\Delta].$ 
\begin{equation*}
\begin{split}
\left\|
b_N(2\Delta) - y_N(2\Delta)
\right\| &\leq \left\|
b_N(\Delta) - y_N(\Delta)
\right\| \\
&\quad {}+ \left\|
\sum_{j = 1}^{M} v_j \int_{0}^{\infty} \left( f_j(b_N(\Delta -s)) - f_j(y_N(\Delta -s)) \right)~d\mu_j(s)
\right\| + \frac{2\zeta}{\sqrt{N\Delta}}
\end{split}
\end{equation*} 
except on a set of measure $2\Delta^{\Delta^{-1}}.$ The bound on the first term is $2\zeta/\sqrt{N\Delta}$ and the bound on the middle term is $2\Delta\nabla_N\zeta/\sqrt{N\Delta}$; the total error is smaller than
\[
\frac{4\zeta}{\sqrt{N\Delta}}+ \frac{2\zeta}{\sqrt{N\Delta}} (\Delta\nabla_N).
\] This argument can now be propagated forward. For any finite $k$, except on a set of $\mathbb{B}_N$ measure smaller than $k\Delta^{\Delta^{-1}}$, the distance between $y_N(k\Delta)$ and $b_N(k\Delta)$ is at most
\[
\frac{2k\zeta}{\sqrt{N\Delta}} + \frac{2\zeta}{\sqrt{N\Delta}} \sum_{j = 1}^{k-1} j (\nabla_N \Delta)^{k -j+1}.
\] The summation on the right can be bounded by $3\zeta(k-1)\nabla_N\Delta/\sqrt{N\Delta}(1 - \nabla_N\Delta)$. Since $k$ is at most $T/\Delta$, the contribution of the summation is small relative to the contribution of the first term. Therefore, we can bound the distance between realizations by $4k\zeta/\sqrt{N\Delta}$.

The proof for the process $x_N$ proceeds analogously. For $N$ large enough
\[
\mathbb{X}_N \left(
\left\|
\delta x_N(0) - \left[
\sum_{j = 1}^{M} v_j \int_{0}^{\infty} f_j(x_N(-s))~d\mu_j(s)
\right]\Delta
\right\| > \frac{\zeta}{\sqrt{N\Delta}}
\right) \leq e^{-\Delta^{-2}/2}
\] from where it follows that 
\[
\mathbb{X}_N \left(
\|\delta x_N(0) - \delta y_N(0)\| > \frac{\zeta}{\sqrt{N\Delta}} 
\right)\leq e^{-\Delta^{-2}/2}.
\] This bound can then be propagated as before. 
\end{proof}


\begin{thebibliography}{51}
\bibitem [{dSS()}]{dSSA_note}%
  \BibitemOpen
  \href@noop {} {\enquote {\bibinfo {title} {While such an assumption is not
  needed for the d{SSA}, it is biochemically realistic and used in our
  mathematical analysis.}}\ }\BibitemShut {NoStop}%
\end{thebibliography}
\end{document}